\documentclass[12pt]{article}
\usepackage{amsfonts}
\usepackage{amsmath}
\usepackage[margin=1in]{geometry}
\usepackage{mathtools}
\usepackage{amsthm}
\usepackage{bbm}
\usepackage{amssymb}
\usepackage{natbib}
\usepackage{graphicx}
\usepackage{caption}
\usepackage{subcaption}
\usepackage{tikz}
\usepackage{setspace}
\usepackage[hidelinks,hypertexnames=false]{hyperref}

\DeclareMathOperator{\supp}{supp}
\DeclareMathOperator*{\argmax}{arg\,max}

\newtheoremstyle{break}
{\topsep}{\topsep}{\itshape}{}{\bfseries}{}
{\newline}{\thmname{#1}\thmnumber{ #2}\thmnote{ (#3)}}
\theoremstyle{break}
\newtheorem{theorem}{Theorem}
\newtheorem{corollary}{Corollary}
\newtheorem{lemma}{Lemma}
\newtheorem{proposition}{Proposition}
\newtheoremstyle{breakdefin}
{\topsep}{\topsep}{\normalfont}{}{\bfseries}{}
{\newline}{\thmname{#1}\thmnumber{ #2}\thmnote{ (#3)}}
\theoremstyle{breakdefin}

\theoremstyle{remark}

\begin{document}

\title{Screening and Segmenting: \\ A Consumer Surplus Perspective}
\author{Dirk Bergemann \and Tibor Heumann \and Michael C. Wang\thanks{Bergemann: Department of Economics, Yale University (email: dirk.bergemann@yale.edu); Heumann: Instituto de Econom\'{i}a, Pontificia Universidad Cat\'{o}lica de Chile (email: tibor.heumann@uc.cl); Wang: Department of Economics, Yale University (email: michael.wang.mcw75@yale.edu). The order of authors is alphabetical.  An early version of this paper working in a more limited setting appeared
under the title \textquotedblleft A Unified Approach to Second and Third
Degree Price Discrimination.\textquotedblright\ We acknowledge financial
support from NSF grants SES-2001208 and SES-2049744 and ANID Fondecyt
Regular 1241302. We acknowledge many productive suggestions from Marzena
Rostek (Co-Editor) and three anonymous referees. We have benefited from many
conversations and related joint work with Ben Brooks and Stephen Morris. We
thank Eugenio Miravete and Justin Johnson for many helpful suggestions. We
thank David Wambach for excellent research assistance.}}

\maketitle

\begin{abstract}
We analyze consumer surplus when a monopolist can
adjust both prices and product qualities across segments, engaging in
second- and third-degree price discrimination simultaneously. We
characterize the consumer-optimal segmentation and show that it has a
striking structure: consumers with the same value receive the same quality
in every segment, though prices differ. Under mild conditions, any
segmentation harms consumers if and only if demand is sufficiently more elastic than supply. Hence, potential benefits for consumers depend
critically on demand and supply elasticities. These findings have implications
for regulatory policy regarding price discrimination and market segmentation.

\vspace{0.5\baselineskip} \noindent \textsc{JEL Classification:} D42, D83,
L12

\end{abstract}

\newpage

\section{Introduction}

\subsection{Motivation and Results}

Market segmentation and product screening are two of the most extensively
studied forms of price discrimination, yet their interaction remains poorly
understood. A firm that segments its market into distinct groups can tailor
not only prices but entire product menus to each group, combining
third-degree price discrimination---different pricing for different
segments---with second-degree price discrimination---screening consumers
through quality-differentiated menus within each segment. This combination
is pervasive in practice: streaming services offer differentiated
subscription tiers to behaviorally segmented audiences, airlines pair yield
management across customer classes with fare-class menus within each class,
and digital platforms simultaneously personalize prices and adjust product
offerings based on consumer data. The central question of this paper is:
when does such combined price discrimination benefit consumers?

The welfare implications of market segmentation have long interested
economists. Third-degree price discrimination---charging different prices to
distinct market segments---may either benefit or harm consumers. A similar
ambiguity exists for second-degree price discrimination, where sellers
screen consumers through quality-differentiated product menus. However, most
research has analyzed these practices in isolation, leaving open the
question of how they interact when deployed simultaneously, as is
increasingly common in practice.

We provide a complete characterization of the consumer-optimal market
segmentation when a monopolist sells goods of varying quality to privately
informed buyers. In our model, the market is partitioned into arbitrary
segments (or submarkets), each characterized by its own distribution of
buyer values, subject only to the requirement that the segments aggregate to
the original market. Taking this segmentation as given, the seller offers a
profit-maximizing screening menu within each segment. We then study the
segmentation that maximizes aggregate consumer surplus.

The consumer-optimal segmentation has a striking structure. Theorem \ref%
{thm:main} establishes that buyers with the same willingness to pay receive
the same quality in every segment even though the monopolist could offer
them different qualities in different segments. Prices, however, vary across
segments: identical buyers may pay different amounts for the same product
across segments. This result is obtained through a dramatic simplification
of the market segmentation problem. The original optimization is over
segmentations---distributions over distributions of values---an inherently
infinite-dimensional and unwieldy object. We can reduce this to maximizing a
single functional, the expected \emph{local information rent}, over a single 
\emph{inverse hazard rate function}, subject to a \emph{majorization
constraint}. The local information rent captures the contribution of buyers
of value $v$ to aggregate consumer surplus as a function of the inverse
hazard rate at that value. The majorization constraint characterizes the
precise limits on how segmentation can redistribute buyers across segments.
This reduction is the paper's central methodological contribution, and it
relies on a novel application of majorization techniques to the space of
inverse hazard rates.

A second consequential finding concerns the conditions under which segmentation
benefits consumers at all. Under a regularity condition on demand and supply, we identify a sharp threshold determined by
demand and supply elasticities (Corollary \ref{cor:no-seg2}). When
aggregate demand is sufficiently elastic relative to the supply, no segmentation can
improve consumer surplus. In such markets, the monopolist's screening
already serves consumers well enough that any redistribution of buyers
across segments only makes consumers worse off. When this condition fails,
the consumer-optimal segmentation takes a particularly simple form (Theorem %
\ref{thm:conv-seg}): it consists of \textquotedblleft convex
segments\textquotedblright\ with nested interval supports that shrink from
below. Moreover, the optimal pricing differs only by a fixed
fee, while the quality menu remains uniform across segments.

These results delineate precisely how the multi-product setting departs from
the unit-demand environment of \cite{bebm15}. There, segmentation can always
achieve efficiency; here, the scope for beneficial segmentation depends on a
quantifiable interaction between the demand and cost primitives. Our results
also clarify the relationship to \cite{hasi23}, who showed that in generic
markets with a finite number of goods, some segmentation always improves
consumer surplus. We show that as the quality space becomes
richer---approaching a continuum---the gains they identify can vanish
entirely, and we provide the precise conditions under which this occurs.

These results have important implications for competition policy and
regulation of price discrimination practices. They suggest that blanket
restrictions on market segmentation may harm consumers by preventing
welfare-enhancing price discrimination, while highlighting specific market
conditions where segmentation is more likely to be beneficial. The findings
also inform ongoing debates about big data and personalized pricing by
showing how consumer heterogeneity and cost structures interact to determine
the nature and effects of optimal segmentation strategies.

The paper is organized as follows. Section \ref{sec:setup} introduces our
model of nonlinear pricing with market segmentation. We consider a seller
who can engage simultaneously in second- and third-degree price
discrimination. Our model consists of a monopolist that offers goods of
varying quality to a continuum of buyers. The willingness-to-pay (value) is
private information to each buyer, and the seller only knows the
distribution of values. The seller may segment the market into submarkets,
each with its own distribution of values, subject only to the condition that
the distribution of values across all submarkets must conform to the
aggregate market. The monopolist offers an optimal pricing scheme in each
submarket.

In Section \ref{sec:ex}, we consider a simple binary value environment,
which illustrates the main concepts we will use throughout the paper. In the
screening problem, the consumer surplus is the information rent. The
contribution of each value to the consumer surplus is identified by the
product of the inverse hazard rate and the allocation. We provide a
convenient representation of this contribution as a function that depends
only on the inverse hazard rate, we refer to this function as the \emph{%
local information rent}. In the binary value environment, valuable
segmentations have the effect of lowering the inverse hazard rate, by \emph{%
concentrating} low-value buyers into a single segment, leaving them missing
in the remaining segments. The maximum consumer surplus corresponds to the
highest local information rent achievable with concentration.

Section \ref{sec:opt-seg} then extends this analysis to the general
environment. With more than two values, the designer has a second tool
alongside concentration: \emph{dilution}, which spreads buyers of a given
value across more segments and so \emph{raises} their inverse hazard rate.
Concentration and dilution are linked---diluting one value requires
concentrating higher values---and our first result (Proposition \ref{prop:avg-h}%
) characterizes exactly what they can jointly achieve. The key condition turns out to
be a familiar one: a \emph{majorization} constraint on inverse hazard rates.

Next, we introduce a class of segmentations, called \emph{uniform
segmentations}, in which the inverse hazard rate of every buyer of the same
value is equalized across segments. While a significant restriction, this
class of segmentations is still quite rich. Nonetheless, we show that the
optimal segmentation must lie within this class. Our main result, Theorem %
\ref{thm:main}, fully characterizes the consumer surplus attained by the
consumer-optimal segmentation. Theorem \ref{thm:main} shows that the value
of segmenting is found by maximizing the expected local information rent
over all inverse hazard rates satisfying the majorization constraint we
found previously. Interestingly, while the original problem consists of a
maximization over segmentations, which are distributions over distributions
of values, Theorem \ref{thm:main} yields a maximization over a single
distribution of values. Theorem \ref{thm:main} thus provides a remarkable
simplification of the original problem.

In Section \ref{sec:conv-seg}, we introduce a symmetric regularity condition---log-concavity---on demand and supply under which the consumer-optimal
segmentation (Theorem \ref{thm:conv-seg}) is easy to characterize. The  resulting
segmentation has the feature that the demand elasticity of low-value buyers
is increased to a fixed level determined by the cost function, while the
segments themselves have support shrinking from the bottom. We refer to
these as \emph{convex segmentations} because the constituent segments are
characterized by their nested convex interval supports, and their
construction is a generalization of the direct segmentations of \cite{bebm15}%
. The particularly simple characterization of the optimal segmentation under
these conditions provides more insight into the conditions under which zero
segmentation is optimal, i.e., any segmentation reduces consumer surplus
(Corollary \ref{cor:no-seg2}). When the model is specialized further to
power cost functions, we uncover a tight connection between supply and
demand elasticities in the optimal segmentation. Section \ref{sec:gen-exts}
then discusses how our methodology can be generalized to other settings,
including the case of discrete distributions, finding all Pareto-efficient
outcomes, and an application to asset trading with adverse selection.
Section \ref{sec:con} concludes. The appendix collects proofs omitted in the
main text.

\subsection{Related Literature}

Our results are related to a large literature on price discrimination,
beginning with \cite{pigo20} and now encompassing a wide range of research
on the output and welfare implications of price discrimination, such as \cite%
{robi33,schm81,vari85,agcv10,cowa12} and \cite{bebm15}. We use the classic
model of second-degree price discrimination via quality and quantity
differentiated products first presented by \cite{muro78} and \cite{mari84}. 
\cite{jomy03} consider a model of second degree price discrimination under
monopoly and duopoly and provide conditions under which the monopolist will
offer a single good.

The problem of extending the results of \cite{bebm15} to a multi-product
setting was analyzed earlier by \cite{hasi22,hasi23}. \cite{hasi22} show
that when the optimal menu in the aggregate market consists of a menu of
more than one item, thus a screening menu, then the consumer surplus
maximizing allocation cannot attain the Pareto frontier and hence the full
surplus triangle cannot be obtained. Based on this insight, they provide a
sufficient condition when the full surplus triangle will be attained, namely
when all value distributions lead to the efficient single item menu. By
contrast, we provide necessary and sufficient conditions for there to be a
segmentation that improves consumer surplus and provide the value of the
consumer surplus maximizing allocation, whether it is efficient or not. \cite%
{hasi23} show that in generic markets there is always a segmentation
relative to the single aggregate market that improves consumer surplus,
however they do not provide properties of the consumer-optimal segmentation.
In contrast to \cite{hasi22,hasi23}, who work with a finite number of
products, we allow for (but do not require) a continuum of qualities.
Therefore, the results in \cite{hasi23} do not generally apply to our
setting, and we find a large class of markets in which zero segmentation is
optimal for consumers, in addition to completely solving for the
consumer-optimal segmentation. We view our work as complementary to \cite%
{hasi23} in the sense that, with a continuum of qualities, the improvement
they identify sometimes becomes negligible, and we provide precise
conditions under which this does (or does not) occur.

Finally, the representation of consumer surplus through the inverse
hazard rate, central to our analysis, has not played this role in the earlier
nonlinear-pricing literature. The inverse hazard rate is of course not new: it
converts values into virtual values in \cite{muro78}, \cite{mari84}, and \cite%
{myer81}. There, however, it is an exogenous summary of the value distribution,
used to characterize the seller's optimal menu. The \emph{joint}
segmentation-and-screening problem promotes it from a primitive to an
instrument: segmentation chooses the within-segment distribution, and hence the
local inverse hazard rate subject to the majorization constraint, while
screening makes the allocation respond to it. The inverse hazard rate then
carries both the extensive margin (the mass of higher values who earn rent) and
the intensive margin (the allocative distortion through the virtual value). This
dual role requires a continuum of qualities together with endogenous
segmentation, a combination absent from both the single-market screening
literature and the finite-good analyses of \cite{hasi22,hasi23}.

\section{Model}

\label{sec:setup}

\paragraph{Payoffs and Menu Pricing}

A monopolist produces vertically differentiated products of varying quality $%
q\in \mathbb{R}_{+}$ at cost $c(q)$, where $c:\mathbb{R}%
_{+}\rightarrow \mathbb{R}_{+}$ is a thrice differentiable, strictly
increasing, and strictly convex function with $c^{\prime }(0)=0$. There is a
continuum of buyers, each characterized by a value $v\in V=[\underline{v},%
\overline{v}]\subset \mathbb{R}_{+}$ that is privately known. A buyer with
value (type) $v$ who purchases quality $q$ at price $p$ obtains net
utility $vq-p$.

The monopolist posts a menu of prices $p:\mathbb{R}_{+}\rightarrow \mathbb{R}%
_{+}$, or \emph{tariff}, which specifies a price $p(q)$ for each quality level $q$. By default, every menu includes the option to buy zero quality at zero price, or $p(0)=0$. (We use $p$
both for the function (menu) and for a particular item price $p$. The
context will identify the relevant interpretation.) Given a menu $p$, a
buyer chooses the quality that maximizes net utility: 
\begin{equation*}
U(v,p)\equiv \max_{q\in \mathbb{R}_{+}}\big[vq-p(q)\big],
\end{equation*}%
and the corresponding quality choice is denoted by 
\begin{equation*}
q({v})\equiv \underset{q\in \mathbb{R}_{+}}{\arg \max }\big[vq-p(q)\big].
\end{equation*}%
If multiple qualities maximize the utility of the buyer, ties are broken in
favor of the seller. The profit of the seller from a buyer with value $v$
when offering menu $p$ is: 
\begin{equation*}
\Pi (v,p)\equiv p(q({v}))-c(q({v})).
\end{equation*}

\paragraph{Market}

A \emph{market} $m\in \Delta V$ is a probability distribution over values $V$%
. $F_{m}$ denotes the cumulative distribution function (cdf) associated with 
$m$ (right-continuous by convention), $f_{m}(v)$ the density where $F_{m}$
is absolutely continuous, and $a_{m}(v)$ the mass of any atom at $v$.

In a given market $m$, the profit-maximizing menu $p_{m}$ solves: 
\begin{equation*}
p_{m}\equiv \argmax_{p(q)}\ \int_{\underline{v}}^{\overline{v}}\Pi (v,p)\
dF_{m}(v).
\end{equation*}%
If there are multiple optimal price menus, we select the one that results in
the highest consumer surplus. We pose the problem for the seller directly in
terms of a menu (or tariff), i.e.\ as an indirect rather than a direct
revelation mechanism. By the taxation principle of \cite{gula84}, these two
approaches are equivalent; our arguments are more transparent when stated in
the quantity-price space.

We denote by $q_{m}(v)$ the optimal choice of value $v$ in market $m$ when
the menu of prices is $p_{m}$. The aggregate consumer surplus in market $m$ is given
by: 
\begin{equation}
U(m)\equiv \int_{\underline{v}}^{\overline{v}}U(v,p_{m})\ dF_{m}(v).
\label{eq:cs}
\end{equation}

\paragraph{Segmentations}

We identify the \emph{aggregate market} with superscript \textquotedblleft $%
\ast $\textquotedblright : 
\begin{equation*}
m^{\ast }\in \Delta (V),
\end{equation*}%
which is the distribution of values of all buyers present in the economy.
And all the associated quantities associated to the aggregate market are
identified with a superscript \textquotedblleft $\ast $\textquotedblright .
For example, in the aggregate market, $F^\ast$ is the cdf and $p^{\ast }$ is
the optimal menu of prices. We assume that the aggregate market is absolutely
continuous with respect to the Lebesgue measure with strictly positive density on $[%
\underline{v},\overline{v}]$ and that $F^{\ast }$ is a real-analytic
function. Section \ref{subsec:other-ext} explains how our results change
when $F^{\ast }$ is discrete.

A \emph{segmentation} $\sigma $ is a distribution over markets $\sigma \in
\Delta (\Delta V)$ such that the composition across market segments equals
the aggregate market $m^{\ast }$. Formally, endow $\Delta V$ with the weak
topology, and let $\sigma $ be a Borel probability measure on $\Delta V$
such that 
\begin{equation}
\int_{m\in \Delta V}m\ d\sigma (m)=m^{\ast }.  \label{eq:agg}
\end{equation}%
We write $\operatorname{MPS}(m^{\ast })$ to denote the set of segmentations which
satisfy the aggregation constraint \eqref{eq:agg}, that is, are
mean-preserving spreads of $m^{\ast }$ in $\Delta (\Delta V)$.

Under a segmentation $\sigma $, the monopolist observes each buyer's segment
but not their value, and posts an optimal menu independently within each
segment or submarket. The seller thus engages simultaneously in
second-degree price discrimination (screening within each segment) and
third-degree price discrimination (across segments).

\paragraph{Consumer Surplus Maximization Problem}

Our central question is: which segmentation $\sigma $ of $m^{\ast }$
generates the highest aggregate consumer surplus? 
\begin{equation}
\max_{\sigma \in \operatorname{MPS}(m^{\ast })}\left[ \int_{m\in \Delta V}U(m)\
d\sigma (m)\right] .  \label{eq:u-opt}
\end{equation}%
We are interested in the value of this maximization problem, as well as the
segmentation $\sigma $ that attains the maximum. Two features of this
formulation deserve emphasis. First, the seller is unconstrained in how they
price within each segment. Second, segmentations are unrestricted: any
partition of the aggregate market is permissible, including those with gaps
and atoms in the constituent submarkets, so long as the aggregation
condition (\ref{eq:agg}) is satisfied.

\section{The Binary Value Case\label{sec:ex}}

This section develops the key economic ideas in a simple environment with
two possible values for the buyers.\footnote{This two-value example is purely
illustrative and outside the continuum model of Section \ref{sec:setup}, as a discrete distribution necessarily violates the assumption that $F^{\ast
}$ be absolutely continuous with full support.} The binary setting is rich enough to illustrate the central
trade-off underlying our results---the tension between the \emph{%
extensive} and \emph{intensive margins} of the information rents---and to
introduce the main
analytical tool, the \emph{local information rent function}, which will
carry over essentially unchanged to the general analysis. At the same time,
the setting is simple enough that all the economic forces are transparent.

\paragraph{Screening with Binary Values}

Suppose buyer values $v_{i}$ are drawn from $\{v_{1},v_{2}\}$
with $0<v_{1}<v_{2}<\infty $. We fix an aggregate market $m^*$ where the low value $%
v_{1}$ arises with probability $a(v_{1})$ and the high value $v_{2}$ with
complementary probability $a(v_{2})=1-a(v_{1})$. We initally just consider the aggregate market and hence can omit the index $m^*$ or $^*$.

Consider a profit-maximizing seller who offers a menu $(
(q_{1},p_{1}),(q_{2},p_{2}))$ of quality--price pairs. As is
standard in the screening literature, the binding constraints are
participation for the low value and incentive compatibility for the high
value, which pin down optimal prices as a function of qualities: 
\begin{equation*}
p_{1}=v_{1}q_{1},\ \ \ p_{2}=v_{1}q_{1}+v_{2}(q_{2}-q_{1}).
\end{equation*}%
Hence, the quality upgrade (or increment) $q_{2}-q_{1}$ is priced at $v_{2}$
but the baseline quality $q_{1}$ is priced at $v_{1}$ (we used that
incentive compatibility requires that $q_{1}<q_{2}$). Consequently, only
high-value buyers get any information rent and the consumer surplus $U$ is given by the expected information
rents: 
\begin{equation*}
U=a(v_{2})(v_{2}-v_{1})q_{1}.
\end{equation*}%
Consumer surplus depends on the quality offered to \emph{low-value} buyers,
even though it is \emph{high-value} buyers who capture the rents. This is
because the rents arise from the high value's ability to mimic the low value:
the better the low value's product, the larger the rent the high value
extracts by threatening to \textquotedblleft trade down.\textquotedblright\
The qualities supplied are of course chosen by the seller to maximize
profits. Hence, the seller chooses qualities $(q_{1},q_{2})$ that solve: 
\begin{equation*}
\max_{(q_{1},q_{2})\in \mathbb{R}_{+}^{2}}\left\{
a(v_{1})(v_{1}q_{1}-c(q_{1}))+a(v_{2})(v_{2}q_{2}-c(q_{2}))-a(v_{2})(v_{2}-v_{1})q_{1}\right\} .
\end{equation*}%
The first two terms are the social surplus generated by qualities $%
(q_{1},q_{2})$, while the last term is the information rents earned by
buyers. The quality $q_{2}$ offered to a buyer with high value $v_{2}$
solves the corresponding first order condition:%
\begin{equation*}
v_{2}=c^{\prime }(q_{2}).
\end{equation*}%
As usual, the highest value does not see any distortion in its allocation.
The quality $q_{1}$ offered to a buyer with low value $v_{1}$ solves:%
\begin{equation}
\phi(v_{1})\equiv v_{1}-(v_{2}-v_{1})\frac{a(v_{2})}{a(v_{1})}\leq
c^{\prime }(q_{1}),  \label{eq:vu}
\end{equation}%
where the above expression holds with equality if the \emph{virtual value} $%
\phi(v_{1})$ is positive. Note that the virtual value depends on the value
distribution. To write this condition explicitly, we denote the inverse
of the marginal cost $c^{\prime }(q)$ by $Q$: 
\begin{equation}
Q(v)\equiv \mathbbm{1}[v\geq 0](c^{\prime })^{-1}(v).  \label{eq:q}
\end{equation}%
We refer to $Q$ as the \emph{supply function}, since $Q(v)$ is the quality
that the monopolist would sell to a buyer of value $v$ if the monopolist
were to offer a socially efficient pricing scheme: $p^\prime (q)
=c^{\prime }(q)$. The supply function is 0 if the value $v$ (or
virtual value)\ is negative, captured by the indicator function $%
\mathbbm{1}[v\geq 0]$. We can write \eqref{eq:vu} using the supply $Q$ by
replacing the value $v_{1}$ with the virtual value $\phi(v_{1})$: 
\begin{equation}
q(v_{1})=Q\left( v_{1}-(v_{2}-v_{1})\frac{a(v_{2})}{a(v_{1})}\right)
=Q(\phi(v_{1})).  \label{eq:opt1}
\end{equation}

\paragraph{The Extensive vs.\ Intensive Margin Trade-Off}

The consumer surplus as a function of the distribution in aggregate market is thus 
\begin{equation}
U\equiv a(v_{2})(v_{2}-v_{1})q(v_{1}),  \label{eq:u}
\end{equation}%
and reveals a fundamental tension in how the market composition affects
consumer surplus. Consider what happens as the fraction of low-value buyers $%
a(v_{1})$ varies. Two forces work in opposite directions.

When there are more low-value buyers, the seller has a stronger incentive to
serve them well. A larger $a(v_{1})$ means a larger customer base whose
quality the seller would sacrifice to improve profits. Hence, the quality $%
q_{1}$ offered to the low value rises with $a(v_{1})$, and so does the
per-buyer rent $(v_{2}-v_{1})q_{1}$ enjoyed by each high-value buyer, the 
\emph{intensive margin} of the rent. But, as $a(v_{1})$ rises, the
fraction of high-value buyers---the \emph{extensive margin} of the
information rent---$a(v_{2})=1-a(v_{1})$ falls. Fewer buyers are collecting the rents.
Eventually, this force dominates, and consumer surplus declines.

Figure \ref{fig:cs-1} illustrates this trade-off for a quadratic cost $%
c(q)=q^{2}/2$ with $V = \{1,2\}$. We display the quantity $q(v_{1})$
delivered to the low-value buyers and the resulting expected consumer
surplus $U$ as a function of the probability $a(v_2)$ of high values.
Observe that when $a(v_2)$ is too high, the seller offers nothing to the low-value buyers, and hence the consumer surplus is also zero. As we decrease $%
a(v_2)$, $q(v_1)$ increases and hence so does $U$. But, when $a(v_2)$
gets too low, there are too few high values to benefit from the information
rents generated by $q(v_1)$, and $U$ eventually goes down again.

\paragraph{Binary Segmentation}
We now consider segmentations of the aggregate market $m^{\ast }$. We focus
on a natural class of binary segmentations $\sigma \in \Delta \{m,m^{\prime
}\}$ in which low-value buyers are absent in market $m^{\prime }$:
\begin{equation*}
a_{m}(v_{1})>0\text{ and }a_{m^{\prime }}(v_{1})=0.
\end{equation*}%

\paragraph{Segmentation as Concentration}
Since low-value buyers appear only in market $m$, their concentration must
satisfy%
\begin{equation}
a_{m}(v_{1})\geq a^{\ast }(v_{1})\text{ and }\sigma ({m})=\frac{a^{\ast
}(v_{1})}{a_{m}(v_{1})},  \label{eq:agg-ex}
\end{equation}%
and with equality when all high-value buyers are in market $m$ (no
concentration) and $a_{m}(v_{1})\rightarrow 1$ representing maximal
concentration. The aggregation constraint for the high-value buyers is then 
\begin{equation*}
\sigma (m)a_{m}(v_{2})+\sigma (m^{\prime })a_{m^{\prime }}(v_{2})=a^{\ast
}(v_{2}).
\end{equation*}

In market $m^{\prime }$, where only high-value buyers are present, the
seller charges the gross utility, so no information rents arise. Hence,
consumer surplus will be generated only in market $m$: 
\begin{equation*}
\sum_{x\in\{m,m^\prime\}}\sigma (x)U(x)=\sigma
(m)(v_{2}-v_{1})a_{m}(v_{2})q_{m}(v_{1}).
\end{equation*}%
The effect of concentration is intuitive. By gathering low-value buyers into
a single segment, the designer \emph{raises} their prevalence in this
segment, which increases the quality the seller offers them and hence the
information rents. The cost is that some high-value buyers are stranded in
the other segment, where they earn zero rents. The question is whether the
deeper rents in $m$ compensate for the narrower base of rent-earning buyers.

\paragraph{The Inverse Hazard Rate}

To analyze this trade-off more cleanly, observe that the market composition
affects consumer surplus entirely through a single statistic: the inverse
hazard rate of the low value, 
\begin{equation*}
h_{m}(v_{1})\equiv (v_{2}-v_{1})\frac{a_{m}(v_{2})}{a_{m}(v_{1})}.
\end{equation*}%
This quantity has a direct economic interpretation: it is the mass of
high-value buyers \emph{per unit of low-value buyers}, scaled by the
difference in values $v_{2}-v_{1}$. The inverse hazard rate captures both
the intensive margin of distortion (through the virtual value):%
\begin{equation*}
\phi _{m}(v_{1})\equiv v_{1}-h_{m}\left( v_{1}\right),
\end{equation*}%
and the extensive margin of rent-earning buyers (since it is proportional to the
high-to-low ratio). Concentration of low-value buyers \emph{lowers} the
inverse hazard rate: making the low value more prevalent reduces the
high-to-low ratio. This is the only tool available in the binary setting:
the segmentation designer can decrease $h_{m}(v_{1})$ relative to $h^{\ast
}(v_{1})$, but never increase it.

\paragraph{The Local Information Rent}

We can now express the consumer surplus $U$compactly using the \emph{local
information rent} function: 
\begin{equation}
u(v,h)\equiv h\cdot Q(v-h).  \label{eq:su}
\end{equation}%
We refer to $u(v,h)$ as the local information rent. Observe that $h$
plays two roles, both captured by the product $h\cdot Q(v-h)$. First, $h$ pins
down the allocation to value $v$: the seller serves $v$ as if its value were the
virtual value $v-h$, so it receives quality $Q(v-h)$, decreasing in $h$. Second,
the leading $h$ is proportional to the (scaled) mass of higher-value buyers who
collect a rent off that allocation, per unit of value-$v$ buyers. Raising $h$
thus widens the base of rent-earning buyers while shrinking the allocation they
earn rent on---the extensive-intensive margin trade-off, compressed into the
single scalar $h$. For the purpose of this example, we will fix $%
v=v_{1}$, and so the relevant variable is the inverse hazard rate $h=h(
v_{1})$. Using the newly introduced notation, we have that total
consumer surplus is: 
\begin{equation*}
\sum_{x\in\{m,m^\prime\}}\sigma (x)U(x)=a^{\ast
}(v_{1})u(v_{1},h_{m}(v_{1})).
\end{equation*}
Figure \ref{fig:cs-2} plots the local information rent function in the same
example, where costs are quadratic and $V = \{1,2\}$.

\begin{figure}[th]
\centering
\begin{subfigure}[b]{0.48\textwidth}
    \centering
    \includegraphics[width=\textwidth]{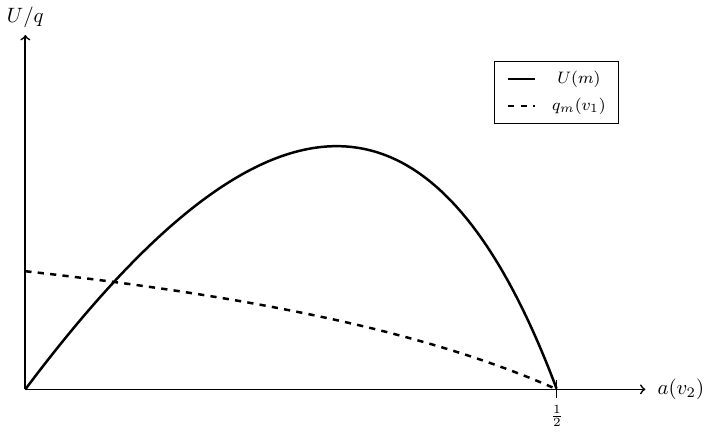}
    \caption{$U$ and $q_m(v_1)$ as a function of $a(v_2)$.}\label{fig:cs-1}
\end{subfigure}
\hfill 
\begin{subfigure}[b]{0.48\textwidth}
    \centering
    \includegraphics[width=\textwidth]{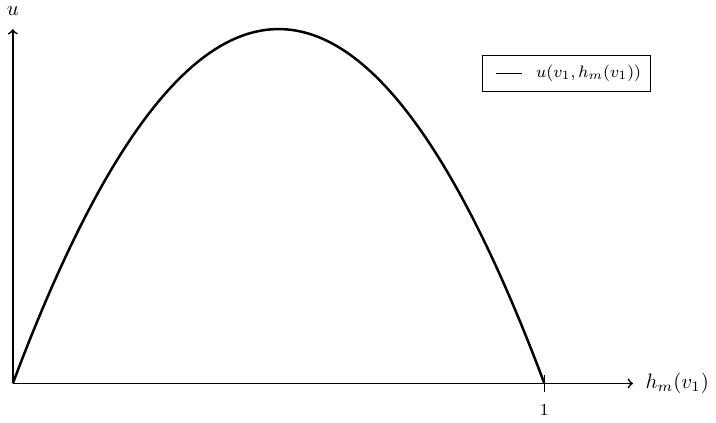}
    \caption{$u(v_1,h)$ as a function of $h_m(v_1)$.}
    \label{fig:cs-2}
\end{subfigure}
\caption{Two ways of finding consumer surplus with $V = \{1,2\}$ and $c(q) =
q^2/2$.}
\label{fig:cs}
\end{figure}

Now the only constraint on the segmentation is given by the constraint %
\eqref{eq:agg-ex}. Namely, the segmentation will increase the concentration
of low-value buyers in market $m$ relative to the aggregate market.
Equivalently, we have that: 
\begin{equation}
h_{m}(v_{1})\leq h^{\ast }(v_{1}).  \label{eq:max-b}
\end{equation}%
The maximum consumer surplus that can be attained with this binary
segmentation is thus: 
\begin{equation}
\max_{h\leq h^{\ast }(v_{1})}\left[ a^{\ast }(v_{1})u(v_{1},h)\right] .
\label{eq:optex}
\end{equation}%
The local information rent $u(v,h)$ encodes the extensive-intensive margin trade-off in
a single function. At $h=0$, there are no high-value buyers, so no rents: $%
u(v,0)=0$. As $h$ increases, more high-value buyers benefit from the low
value's allocation, raising total rents. But simultaneously, the virtual
value $v-h$ falls, so the seller distorts quality downward more
aggressively. When $h\rightarrow v$, the distortion is so severe that the
seller excludes the low value entirely: $Q(v-h)=0$, and rents vanish. Between
these extremes, $u(v,h)$ is hump-shaped, attaining its maximum at an
interior value 
\begin{equation*}
\overline{h}(v)\equiv \argmax_{h}\big[u(v,h)\big].
\end{equation*}%
If the information rents are quasi-concave, the solution to the maximization
problem \eqref{eq:optex} is: 
\begin{equation*}
h=%
\begin{cases}
h^{\ast }(v_{1}), & \text{if }h^{\ast }(v_{1})\leq \overline{h}(v_{1}); \\ 
\overline{h}(v_{1}), & \text{if }h^{\ast }(v_{1})>\overline{h}(v_{1}).%
\end{cases}%
\end{equation*}%
In other words, if the mass of high-value consumers is too low, then the
optimal segmentation is zero segmentation; if the mass of high-value
consumers is too high, then some of them are separated into another market.

\paragraph{Preview of the General Analysis}

The analysis of this section is incomplete, as we did not show that the
restricted binary segmentations are optimal (with binary values, they are).
However, even as we move to continuum values and consider arbitrary
segmentations, the expressions for the consumer surplus in terms of expected
local information rents will remain essentially the same.

The binary setting illustrates the core mechanics, but it is special in two
respects. First, the designer can only \emph{concentrate} buyer values (lower$%
\ h)$, never \emph{dilute} them (raise $h$). With a continuum of values,
dilution becomes possible by creating gaps in the support of some segments,
but it requires concentrating adjacent values first to create the gaps. This
interplay between concentration and dilution is governed by a \emph{%
majorization constraint} that generalizes the pointwise constraint $%
h_{m}(v_{1})\leq h^{\ast }(v_{1})$ in (\ref{eq:max-b}). Second, with binary
values there is only a single inverse hazard rate to optimize. With a
continuum of values, the designer must choose an entire function $h(v)$,
balancing rents across all values simultaneously. Despite these added
dimensions, the structure of the solution is remarkably parallel: consumer
surplus is still the expected local information rent, and the
consumer-optimal segmentation is still found by maximizing $u(v,h(v))$ over
feasible inverse hazard rates.

\section{Consumer-Optimal Segmentations\label{sec:opt-seg}}

We now turn to the general environment with a continuum of buyer values and
unrestricted segmentations. Our goal is to show that the problem of finding
a consumer-optimal segmentation---which is, on its face, an optimization
over distributions of distributions---can be reduced to a far simpler
problem: maximizing a single functional over a single function, subject to a
family of linear constraints. We build toward the main result in four steps.

First, we show that the local information rent representation of consumer
surplus from Section \ref{sec:ex} extends to the general setting, so that
consumer surplus in any regular market equals the expected local information
rent (Section \ref{subsec:screening}). Second, we characterize the limits of
what segmentation can achieve by identifying the precise constraints on
inverse hazard rates that any segmentation must satisfy (Section \ref%
{subsec:redist}). The key insight is that the feasible inverse hazard rates
are governed by a \emph{majorization constraint}---a generalization of the
pointwise constraint $h\leq h^{\ast }$ that arose in the binary case. Third,
we introduce a class of segmentations, called \emph{uniform segmentations},
in which every buyer of the same value faces the same inverse hazard rate in
every segment where they appear, and we argue that it is sufficient to
restrict attention to this class (Section \ref{subsec:uniform}). Finally, we
combine these ingredients to state and prove our main result (Section \ref%
{subsec:opt-seg}). All the illustrations in this section correspond to the
optimal segmentation when a seller has two indivisible goods with different
marginal costs; here, we use them only to illustrate the concepts and
results, rather than provide a detailed analysis.

\subsection{\texorpdfstring{From Binary to Continuum: \newline Consumer Surplus as Expected Local Information Rent}{From Binary to Continuum: Consumer Surplus as Expected Local Information Rent}}

\label{subsec:screening}

In the binary-value setting, we showed that consumer surplus is represented
by the local information rent which depends on the inverse hazard rate. This
representation extends naturally to a continuum of values. The key step is
defining the inverse hazard rate for general distributions that may contain
both continuous and discrete components, density and atoms, and the support
of a market may have gaps.\footnote{It is without loss to restrict to distributions of this form, as markets with a singular-continuous component are irregular; see proof of Lemma \ref{lm:csir}.}

For any market $m$ and any value $v$ in the support of $m$, define the \emph{%
increment} at $v$, denoted by $\Delta _{m}(v)$, as the distance between $v$
and the next highest value present in $m$: 
\begin{equation*}
\Delta _{m}(v)\equiv \inf \big\{t\in \mathrm{supp}(m)\mid t>v\big\}-v.
\end{equation*}%
We define the inverse hazard rate in segment $m$ as: 
\begin{equation}
h_{m}(v)\equiv 
\begin{cases}
\frac{1-F_{m}(v)}{f_{m}(v)}, & \text{if }a_{m}(v)=0\text{ and }%
0<f_{m}(v)<\infty ; \\ 
\Delta _{m}(v)\frac{1-F_{m}(v)}{a_{m}(v)}, & \text{if }a_{m}(v)>0; \\ 
0, & \text{if }a_{m}(v)=f_{m}(v)=0.%
\end{cases}
\label{eq:ihr}
\end{equation}%
The definition of the inverse hazard rate has three components. The first
one is when the distribution is continuous at some $v$, the second one is
when there is an atom at $v$, and the third one is when $v$ is not in
support of $m$, and hence there is a gap around $v$. We define the inverse
hazard rate to be 0 if $v\notin \supp(m)$ because when $v$ is not present in
market $m$, it does not generate information rents.\footnote{%
We could have replaced the value of $h_{m}(v)$ with any finite value when $%
v\not\in \mathrm{supp}(m)$ because it will be integrated over a null
density, but we want to keep it finite to avoid multiplying infinity times
zero.} The economic content of this definition is the same as in the binary
case: the inverse hazard rate measures the mass of buyers above $v$ per unit
of value $v$, scaled by the local spacing of the support. When $m$ is
continuous, it reduces to the familiar ratio of the survival function to the
density.

For any market $m$ and any $v\in \supp(m)$, we define the \emph{virtual value%
} as: 
\begin{equation}
\phi_{m}(v)\equiv v-h_{m}(v).  \label{eq:vv-cont}
\end{equation}%
This definition incorporates both the case of continuous and discrete
distributions. As in Section \ref{sec:ex}, we denote by $Q$ the inverse of
the marginal cost function (see \eqref{eq:q}). A market is regular if for
every $v$, the quality that the buyer with value $v$ consumes $q_{m}(v)$
under the profit-maximizing menu is the supply evaluated at the
corresponding virtual value: 
\begin{equation*}
q_m(v)=Q(\phi _{m}(v)).
\end{equation*}
We can restrict attention to segmentations that place positive weight only
on regular markets.

Two implications of regularity recur below. (i) In a regular market, the
profit-maximizing menu involves no ironing: value $v$ consumes exactly $%
q_{m}(v)=Q(\phi _{m}(v))$, which is what makes the local information rent
representation of Lemma \ref{lm:csr} available. When $m$ has full support,
regularity is equivalent to the virtual value $\phi _{m}$ being non-decreasing.
(ii) Regularity requires an atom at the lower endpoint of any gap in the
support. Roughly, if instead the density continued up to the gap, the quality allocation
would jump downward at the top of that density, violating monotonicity.

\begin{lemma}[Consumer Surplus in Irregular Markets]
\label{lm:csir} For every market $m^\prime$, there exists a segmentation $%
\sigma$ of $m^\prime$ into regular markets such that 
\begin{equation}  \label{eq:u-irreg}
U(m^\prime) = \int_{\Delta V} U(m)\ d\sigma(m).
\end{equation}
\end{lemma}

This lemma shows that any irregular market can be segmented into regular
markets while preserving the same profit-maximizing pricing. Our notion of
regularity requires that the ironed virtual values agree with the original
ones (details in proof of Lemma \ref{lm:csir}). When $m$ is full support, it
reduces to having monotone virtual values. Hence, this statement uses no
properties of how consumer surplus is determined, but simply shows that one
can segment irregular markets into regular ones in an inconsequential way.
Thus, for the remainder of the section, we will only consider segmentations
supported on regular markets, which is sufficient to characterize the
maximum consumer surplus achievable with segmentations. It will turn out
that \eqref{eq:u-irreg} can be strengthened to a strict inequality, that is,
segmentations supported on irregular markets are strictly suboptimal; we
relegate the details of this argument to the proof of Theorem \ref{thm:main}
in the appendix.

The \emph{local information rent} is defined the same way as \eqref{eq:su},
which we display again: 
\begin{equation*}
u(v,h)\equiv h\cdot Q(v-h).
\end{equation*}%
As in the binary value case, we can express the consumer surplus in regular
markets as the expected local information rent via the envelope theorem as it is standard.

\begin{lemma}[Consumer Surplus in Regular Markets]
\label{lm:csr} In every regular market $m$, the consumer surplus is given
by: 
\begin{equation}
U(m) = \int_{\underline{v}}^{\overline{v}} u(v,h_m(v))\ dF_m(v).
\label{eq:u-reg}
\end{equation}
\end{lemma}

The local information rent $u(v,h)$ captures the total rents generated by
value $v$ per unit mass of buyers, so when multiplied by the density of
buyers, it gives the total consumer surplus generated by $v$. The two
roles of the inverse hazard rate identified in the binary case carry over
unchanged: the quality consumed by value $v$ is $q_{m}(v)=Q(v-h_{m}(v))$, which
depends on the market only through $h_{m}(v)$, while $h_{m}(v)$ is itself the
mass of higher values benefiting from that allocation per unit of value $v$. This
is why the entire segmentation problem can be recast in the single function $h$:
fixing the inverse hazard rate at each value determines both the allocation and
the mass benefiting from it, and hence consumer surplus. We emphasize
that the local information rents do not identify which values are \emph{%
receiving} the rents, but rather which values are \emph{generating} the
rents. Value $v$ generates rent for higher values because the seller offers
a quality level at price $v$, which in turn generates rent for all higher
values that can purchase this quality level at a price below their own
valuation. Hence, throughout this section, we consider the following
problem: 
\begin{equation}  \label{reg:maxim}
\begin{split}
&\max_{\sigma\in \operatorname{MPS}(m^\ast)}\int_{\Delta V}\int_{\underline
v}^{\overline v} u(v,h_m(v))\ dF_m(v)\ d\sigma(m) \\
&\text{subject to: every }m\in\mathrm{supp}(\sigma)\text{ is regular,}
\end{split}%
\end{equation}
which delivers the same value as \eqref{eq:u-opt}.

\subsection{\texorpdfstring{What Can
Segmentation
Achieve?
\newline
Concentration, Dilution and
Majorization}{What
Can
Segmentation
Achieve? Concentration, Dilution and
Majorization}}

\label{subsec:redist}

Our expression for consumer surplus in regular markets \eqref{eq:u-reg}
invites us to consider how we might change the inverse hazard rate $h$ to
increase local information rents. In general, higher $h$ means more
high-value buyers benefit because it raises the mass of buyers above $v$ per
unit of value $v$ (as measured by the inverse hazard rate). Simultaneously,
however, it means lower quality for value $v$, so distortions and prices
increase with the inverse hazard rate. These opposing forces imply that
local information rents are maximized at an interior inverse hazard rate $%
\overline{h}(v)$, defined by 
\begin{equation}
\overline{h}(v)\equiv \argmax_{h\in \mathbb{R}_{+}}\big[u(v,h)\big].
\label{eq:maxh}
\end{equation}%
If multiple maximums exist, we take the lowest one.

\paragraph{Two Tools: Concentration and Dilution}

The binary-value analysis identified a single tool available to the
segmentation designer: \emph{concentration}, gathering buyers of a given
value into one segment, thereby lowering their inverse hazard rate. With more
values, a second tool becomes available: \emph{dilution}, raising the inverse
hazard rate of a given value by placing that value in segments where
immediately higher values are \emph{missing}. When a given value is spread
out over more markets than the immediately higher values, a gap opens in the
support above value $v$, which increases the increment $\Delta _{m}(v)$ and
hence the inverse hazard rate (\ref{eq:ihr}).

Dilution is economically quite different from concentration. Concentrating a
value is straightforward: place more of those buyers in some segments and
remove them from others. But diluting a value---raising its inverse hazard
rate---requires that some values just above $v$ have been \emph{removed} from
the segment. And those values can only be removed if they have been
concentrated elsewhere. In other words, dilution at one value requires
concentration at adjacent higher values.

\paragraph{The Majorization Constraint}

For this reason, concentration is valuable both as a tool in its own right
to modify the inverse hazard rate, and to enable dilution of lower values.
Our key challenge, then, is to characterize the limits of what inverse
hazard rates are achievable through segmentation. For this purpose, let $%
h:V\rightarrow \mathbb{R}_{+}$ be an arbitrary function. We say $h^{\ast }$ 
\emph{majorizes} $h$ and write $h^{\ast }\succ h$ if: 
\begin{equation}
\int_{v}^{\overline{v}}h^{\ast }(t)dF^{\ast }(t)\geq \int_{v}^{\overline{v}%
}h(t)dF^{\ast }(t),  \label{eq:maj-00}
\end{equation}%
for all $v\in V$. We refer to \eqref{eq:maj-00} as the \emph{majorization}
constraint, as it is closely related to the familiar notion of majorization
(with the distinction that $h$ here need not be monotone in $v$). We define
the slack in the majorization constraint at $v$ as $E_{h}(v)$:
\begin{equation}
E_{h}(v)\equiv \int_{v}^{\overline{v}}h^{\ast }(t)-h(t)dF^{\ast }(t)\geq 0.
\label{eq:maj-0}
\end{equation}%
In Figure \ref{fig:maj-a} we plot the inverse hazard rate of an aggregate
market and an inverse hazard rate that satisfies the majorization constraint
(labeled as $h_\sigma$). We also plot the slack in the majorization
constraint $E$ and observe that its slope changes sign depending on whether
the hazard rate is above or below the aggregate market.

\begin{figure}[th]
\centering
\begin{subfigure}[b]{0.48\textwidth}
    \centering
    \includegraphics[width=0.85\textwidth]{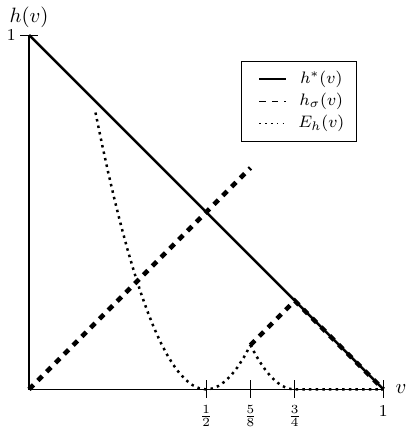}
        \caption{Inverse hazard rate satisfying majorization constraint, and that of the aggregate market.}
    \label{fig:maj-a}
\end{subfigure}\hfill 
\begin{subfigure}[b]{0.5\textwidth}
    \centering
    \includegraphics[width=\textwidth]{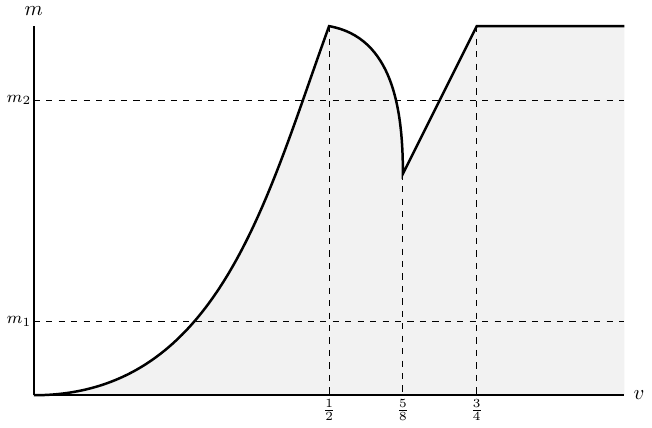}
    \caption{Support of a segmentation implementing $h_\sigma$ shown in Figure \ref{fig:maj-a}.}\label{fig:seg}
\end{subfigure}
\caption{Inverse hazard rates and a segmentation that implements it.}
\label{fig:maj}
\end{figure}

For any segmentation $\sigma $, let $\Sigma (v,m)$ be the induced joint
distribution over values and markets, and $\Sigma _{v}(m)$ the marginal
distribution of $\Sigma $ over markets, conditional on value. Define the
average inverse hazard rate $h_{\sigma }$ as: 
\begin{equation}
h_{\sigma }(v)\equiv \int_{\Delta V}h_{m}(v)\ d\Sigma _{v}(m).
\label{eq:h-avg}
\end{equation}%
This is, in general, not the same as the average value of $h_{m}(v)$ over $%
\sigma (m)$, unless $h_{m}(v)$ is constant across all $m\in \supp(\sigma )$.
If $F_{m}$ were continuous for all $m\in \supp(\sigma )$, this would
simplify to the weighted average value of $h_{m}(v)$,
\begin{equation*}
h_{\sigma }(v)=\int_{\Delta V}h_{m}(v)\frac{f_{m}(v)}{f^{\ast }(v)}\ d\sigma
(m),
\end{equation*}%
where more weight is given to markets that have more buyers of value $v$. The
integral expression \eqref{eq:h-avg} extends this to the case where $\sigma $
contains distributions with atoms.

This notion of average hazard rate formalizes our intuition that diluting
buyers of value $v$ requires gaps to have been created above $v$.
Specifically, for all $v$, 
\begin{equation}  \label{eq:h-sigma}
E_h(v) = \underset{\{m:v\not\in\mathrm{supp}(m)\}}{\int} \Delta_m(v)
(1-F_m(v))\ d\sigma(m).
\end{equation}
The right-hand side (rhs) integrates the ``missing demand'' in each market
where $v$ is missing (that is, $v\not\in\mathrm{supp}(m)$), weighted by the
distance to the nearest value in the support. Since it is weakly positive by
construction, we get that the majorization constraint \eqref{eq:maj-0} is
satisfied.

\begin{proposition}[Inverse Hazard Rates in Segmentations]
\label{prop:avg-h} For any segmentation $\sigma$ of $m^\ast$, $h_\sigma
\prec h^\ast$.
\end{proposition}

In the next section, we will use a partial converse of this result: given $%
h\prec h^{\ast }$, we can construct a segmentation $\sigma $ such that $%
h_{\sigma }=h$, provided that $h$ satisfies certain technical conditions.
Thus, in a strong sense, the majorization constraint captures the maximal
extent to which the redistribution of buyers can affect the inverse hazard
rates, which is itself sufficient for capturing the outcomes of segmentation.

The contrast with the binary case is illuminating. There, the only feasible
change was $h\leq h^{\ast }$ pointwise; the designer could only concentrate
(lower $h$), never dilute (raise $h$). With a continuum of values, the
majorization constraint is \emph{weaker} than the pointwise constraint: it
permits the inverse hazard rate to \emph{exceed} $h^{\ast }$ at some values,
provided it falls sufficiently below $h^{\ast }$ at other (higher) values to
maintain the cumulative inequality. This additional flexibility is precisely
what dilution provides, and it can strictly expand the achievable consumer
surplus.

\subsection{Uniform Segmentations: Equalizing Returns across Segments\label%
{subsec:uniform}}

The local information rent at value $v$ describes the contribution of every
unit of buyers of value $v$ to the consumer surplus in a given segment. If
the same value makes different contributions in different segments, then we
might want to redistribute buyers of this value to wherever they can make the
largest contribution. Thus, we should expect that the optimal segmentation
equalizes the returns of given value $v$ across all segments. But, this
suggests that all segments should have a uniform inverse hazard rate for a
given value $v$, meaning every value generates the same inverse hazard rate
in every segment that the value is present.

Formally, we say that $\sigma$ is a \emph{uniform segmentation} if, for
almost every $m \in \supp(\sigma)$ and $v \in \supp(m)$, 
\begin{equation}  \label{eq:equal-h}
h_m(v) = h_\sigma(v).
\end{equation}
Segments within a uniform segmentation share a common shape along their
shared supports that may differ from the aggregate market. In Figure \ref%
{fig:uniform-seg-a}, we plot two distributions that have the same inverse
hazard rates (which is the one in Figure \ref{fig:maj-a}). The aggregate
market $F^{\ast }$ is drawn only for reference. Segments $%
m_{1}$ and $m_{2}$ are two members of the continuum making up the segmentation,
so they need not---and in general do not---average to $F^{\ast}$;
only the full collection does. Figure \ref{fig:seg} shows how $m_{1}$ and $m_{2}$
sit among all the segments. Market $m_1$ is a
gapless distribution, while market $m_2$ has a gap (flat segment) and an
atom (discontinuous jump) that are in exact proportion so that the discrete
inverse hazard rate is the same as market $m_1$. Even though both
distributions have the same inverse hazard rates (over their overlapping
support) they do look quite different, illustrating the richness that exists
even among uniform segmentations. In Figure \ref{fig:uniform-seg-b}, we plot
the corresponding densities, where we again observe similar shapes in their
supports.

\begin{figure}[th]
\centering
\begin{subfigure}[b]{0.48\textwidth}
    \centering
    \includegraphics[width=\textwidth]{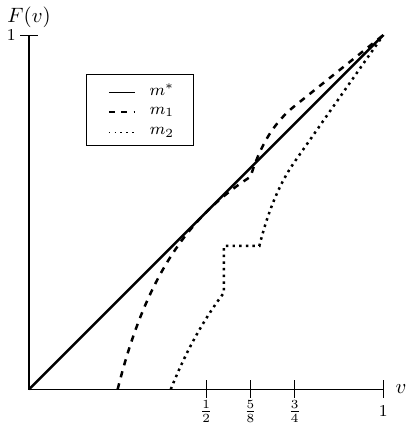}
    \caption{Cdf of aggregate market and segments.}\label{fig:uniform-seg-a}
\end{subfigure}\hfill 
\begin{subfigure}[b]{0.48\textwidth}
    \centering
    \includegraphics[width=\textwidth]{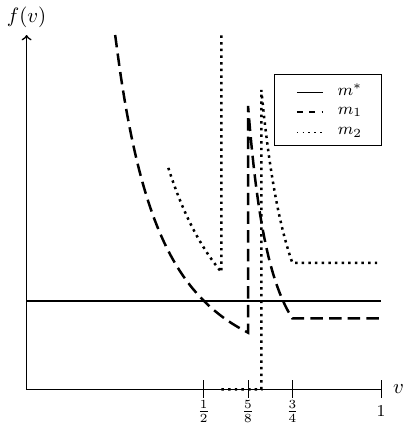}
    \caption{Densities of the same markets.}
    \label{fig:uniform-seg-b}
\end{subfigure}
\caption{Two segments (out of a continuum) in a uniform
segmentation of $\mathcal{U}[0,1]$. The aggregate market is shown for reference
only; the two segments alone do not average to it.}
\label{fig:uniform-seg}
\end{figure}

A uniform segmentation $\sigma $ will consist only of regular markets if and
only if $v-h_{\sigma }(v)$ is non-decreasing. Every uniform segmentation $%
\sigma $ that consists only of regular markets generates consumer surplus 
\begin{equation*}
\int_{\Delta V}\int_{\underline{v}}^{\overline{v}}u(v,h_{m}(v))\ dF_{m}(v)\
d\sigma (m)=\int_{\underline{v}}^{\overline{v}}u(v,h_{\sigma }(v))\ dF^{\ast
}(v),
\end{equation*}%
where here we are using expression \eqref{eq:u-reg} to express the consumer
surplus as the expected local information rents and \eqref{eq:equal-h} to
replace the hazard rates. We can maximize this expression over $h_{\sigma
}\prec h^{\ast }$ to obtain an upper bound on how much consumer surplus a
uniform segmentation supported only on regular markets can generate: 
\begin{equation}
\sup_{h\prec h^{\ast }}\int_{\underline{v}}^{\overline{v}}u(v,h(v))\
dF^{\ast }(v)  \label{eq:ub}
\end{equation}%
In order for this upper bound to be tight, there must be an $h$ which
achieves the maximum of \eqref{eq:ub}, and a uniform segmentation
implementing this $h$ which is supported only on regular markets. All these
conditions are satisfied.

\begin{proposition}[Optimal Uniform Segmentation]
\label{prop:unif-seg-ub} There exists an $h$ that attains the supremum of %
\eqref{eq:ub}, and a uniform segmentation $\sigma$ implementing it supported
only on regular markets.
\end{proposition}

The propositions show that the relevant (optimal) inverse hazard rate that
satisfies the majorization constraint can be implemented by a uniform
segmentation. We prove the existence of such segmentations constructively,
where markets are indexed by $m\in [0,1]$ and a function $s(v)$ marks the
upper boundary of the support: only market indices below $s(v)$ contain
value $v$. Hence, markets have nested supports. In Figure \ref{fig:seg} we
illustrate the support of the segmentation that implements inverse hazard
illustrated in Figure \ref{fig:maj-a}; the $x-$axis represents values $v\in[%
0,1]$, $y$-axis represent a market index $m\in[0,1]$, and the shaded area is
the pairs of value-market ($v$,$m$) that have positive support. In the $y$%
-axis, we show where markets $m_1,m_2$ are (illustrated in Figure~\ref%
{fig:uniform-seg-a}), and see that the gray shaded area marks the support of
these distributions. Whenever the upper boundary of the shaded area has a
positive (negative) slope, we get concentration (dilution), corresponding to
the case where the inverse hazard rate is smaller (larger) than in the
aggregate market (as illustrated in Figure \ref{fig:uniform-seg-b}). At the
upper boundary of the segment with negative slope (between 1/2 and 5/8), a
market has an atom, which significantly shrinks the mass of values available
to distribute across other markets; hence, even if a value is not
distributed across all the markets, the inverse hazard rate is higher than
in the aggregate market when the slope of the upper boundary is negative.

The existence of an optimal $h$ is non-trivial to prove, but the arguments
are technical, so we relegate them to the appendix. From an economically
substantive perspective, the more important aspect of this proposition is
that the optimal $h $ always satisfies regularity: $v-h(v)$ is
non-decreasing (alternatively, $h^{\prime }(v)\leq 1$). This means the
segments in the optimal uniform segmentation are themselves regular, which
has a clear economic logic.

Recall from the binary case that the local information rent $u(v,h)$ is
hump-shaped in $h$: increasing $h$ raises the extensive margin of rents but increases
distortions. A necessary condition for $h$ to be optimal is that the
marginal local information rent is that 
\begin{equation}
\frac{\partial }{\partial h}u(v,h(v))  \label{eq:d-gaph}
\end{equation}
is non-decreasing in $v$. If this failed---if the marginal value of raising $%
h$ were higher for some low value than for a nearby high value---the designer
could concentrate the high value (lowering their $h$) to create slack for
diluting the low value (raising their $h$), generating a net increase in
consumer surplus. Now, because $u$ is ``not too supermodular'' (in a sense
formalized in the appendix), this condition can be used to show that higher
values cannot be associated with too much higher hazard rates (formally, $%
h^{\prime }(v)\leq 1$ in any solution to \eqref{eq:ub}).

\subsection{The Main Result: The Value of Segmentation\label{subsec:opt-seg}}

Proposition \ref{prop:unif-seg-ub} means that the maximum of \eqref{eq:ub}
is a lower bound on the consumer surplus achievable by a segmentation. To
show that this bound is tight, we need to rule out that higher values can be
achieved by non-uniform segmentations. We do this with a short technical
argument, which tracks the intuition provided at the start of Section \ref%
{subsec:uniform}: if buyers of the same value are contributing
differentially to the local information rent, they should be reallocated to
segments where they contribute the most. If $u$ were concave, this would
follow from Jensen's inequality; although $u$ is not concave, we show that
the second-order condition for the optimal $h$ is strong enough to apply
this argument nonetheless.

\begin{theorem}[Value of Segmentation]
\label{thm:main} The maximum consumer surplus attained by segmentation is: 
\begin{equation}  \label{eq:main-thm}
\max_{\sigma \in \operatorname{MPS}(m^\ast)} \left[\int_{\Delta V} U(m)\ d\sigma(m) %
\right] = \max_{h\prec h^\ast}\ \int_{\underline v}^{\overline v}u(v,h(v))\
dF^\ast(v).
\end{equation}
Furthermore, every segmentation achieving this value is a uniform
segmentation implementing some inverse hazard rate $h$ which solves the rhs of %
\eqref{eq:main-thm}.
\end{theorem}

Theorem~\ref{thm:main} is the central result of the paper. We highlight
several aspects of its content and implications.

\paragraph{Reduction of Complexity}

The original problem, given by the lhs of \eqref{eq:main-thm}, requires
maximizing over segmentations $\sigma \in \Delta (\Delta V)$: distributions
over the infinite-dimensional space of all probability distributions on $V$.
The reduced problem, given by the rhs of ~\eqref{eq:main-thm}, requires maximizing a single
functional over a single function $h:V\rightarrow \mathbb{R}_{+}$, subject
to the family of linear majorization constraints~\eqref{eq:maj-00}. This is
a standard optimal control problem, amenable to Karush-Kuhn-Tucker (KKT)\
methods and often yielding closed-form solutions. The reduction is thus both
conceptually clarifying and computationally powerful.

\paragraph{Relation to Extreme Points and Majorization}

Theorem \ref{thm:main} contrasts with the extreme-point approach to
majorization problems \citep{klms21}. The original problem \eqref{eq:u-opt} maximizes an objective which is linear in $\sigma$ over the convex set
$\operatorname{MPS}(m^{\ast})$, so an optimum is attained at an extreme point (an
``extremal'' segmentation into regular markets). However, the set $\operatorname{MPS}(m^\ast)$ is taken over $\Delta(\Delta V)$, not $\Delta V$, and hence lacks the usual tractable one-dimensional reduction to majorization. By contrast, in our simplified problem \eqref{eq:main-thm}, the objective in terms of inverse hazard rates $\int_{\underline{v}}^{\overline{v}}u(v,h(v))\,dF^{\ast }(v)$
is \emph{not} linear (nor convex) in $h$, and the optimal inverse hazard rate is therefore generically interior. Thus, the
standard tools for maximizing a convex functional over a majorized set
do not apply.

\paragraph{Uniform Quality across Segments}

Since $h(v)$ is the same in every segment containing value $v$, and
the quality allocation in a regular market is $q_{m}(v)=Q(v-h_{m}(v))$, it
follows that every buyer of value $v$ consumes the \emph{same quality} in
every segment where they appear. This is a strong and perhaps
counterintuitive property: even though the monopolist could offer different
products to the same buyer in different segments, the consumer-optimal
segmentation ensures uniform quality provision. What varies across segments
is only the \emph{price} for the quality.

\paragraph{Comparison with the Binary Case}

In the binary setting, the consumer surplus maximization reduced to $\max_{h
\leq h^\ast(v_1)} u(v_1, h)$, a single-variable optimization with a
pointwise constraint. Theorem~\ref{thm:main} reveals that the natural
generalization is \emph{not} the pointwise problem $\max_{h \leq h^\ast}
\int u(v,h(v)) \, dF^\ast(v)$, but rather the majorization-constrained
problem~\eqref{eq:main-thm}. The majorization constraint is strictly weaker
than the pointwise constraint, because it permits the designer to raise the
inverse hazard rate at low values by concentrating higher values. This
additional flexibility can strictly improve consumer surplus, and it arises
from the dilution mechanism that was absent in the binary case. In many
environments, however---including those studied in Section~\ref{sec:conv-seg}%
--- dilution is not optimal. When dilution is not used, the solution
coincides with the one obtained if, instead of having the majorization
constraint, the inverse hazard rate were a pointwise bound on the markets'
inverse hazard rate $h(v)\leq h^\ast(v)$.

\paragraph{Uniqueness}

The second part of Theorem~\ref{thm:main} establishes that \emph{every}
optimal segmentation is uniform. We prove this by showing that irregular
markets are strictly suboptimal (hence, strengthening Lemma~\ref{lm:csir}),
and non-uniform segmentations are also strictly suboptimal by the concavity
argument above. Thus, the optimal segmentation is essentially unique in
terms of its economic content (the inverse hazard rate function $h$ and the
resulting quality allocations), even though multiple uniform segmentations
may implement the same inverse hazard rate $h$.

\section{Convex Segmentations\label{sec:conv-seg}}

Theorem \ref{thm:main} reduces the segmentation problem to maximizing a single functional
over a single function---a remarkable simplification. Yet the solution it
delivers can still be intricate: optimal inverse hazard rates may require
segments with gaps and atoms, even in simple environments with finitely many
goods
. In this section, we
show that a widely used regularity condition---on the demand as well as on
the supply function---yields a solution of striking simplicity.

The optimal segmentation turns out to consist of ``convex'' segments: nested
intervals all sharing the same upper bound $\overline{v}$, differing only in
how far down in the support of the value distribution they reach. These are segmentations
where, in contrast to Figure \ref{fig:seg}, the slope of the upper boundary
is always increasing. Furthermore, the shape of all optimal segments can be
expressed in terms of a single transparent interaction between demand and
supply elasticities.

We develop these results in three steps. First, we introduce the two
regularity conditions and interpret them economically. Second, we
characterize the optimal inverse hazard rates and the associated pricing
structure. Next, we specialize to iso-elastic costs and uncover a tight
connection between supply and demand elasticities. Finally, we use these
tools to provide a sharp, easily interpretable condition under which no
segmentation can help consumers.

\subsection{Log-Concave Demand and Supply\label{subsec:con}}

The analysis in this section relies on a symmetric assumption for demand and
supply. We assume that the aggregate demand and supply function are
log-concave. Specifically, the demand in the aggregate market, the survival
function $1-F^{\ast }(v)$, is assumed to be log-concave in $v$: 
\begin{equation}
\frac{d^{2}}{dv^{2}}\log (1-F^{\ast }(v)) \leq 0.  \label{eq:mhr}
\end{equation}%
This is equivalent to the monotone (inverse) hazard rate condition: the
assumption that $h^{\ast }(v)$ is non-increasing in $v$. It is satisfied by
most commonly used distributions, including the uniform, normal, logistic,
and exponential families (see \cite{babe05} for a comprehensive catalog).

Similarly, the supply function is also assumed to be log-concave in $v$:
\begin{equation}
\frac{d^{2}}{dv^{2}}\log (Q(v))\leq 0.  \label{eq:conv-mc}
\end{equation}%
We can also provide this condition in terms of the cost function: 
\begin{equation*}
\frac{c^{\prime \prime \prime }(q)q}{c^{\prime \prime }(q)}\geq -1.
\end{equation*}%
This condition requires the marginal cost to be convex or not too concave,
and accommodates a wide range of cost structures, including the entire
family of power cost function studied in Section \ref{subsec:iso-cost}
(which includes the single unit demand model as a limiting case).

What do these conditions have in common? They require that either demand or
supply elasticity does not grow faster than proportional with price. The
log-concavity condition appears prominently in the pass-through literature
in international trade and industrial organization, see \cite{bupf83}, \cite%
{wefa13} and \cite{mrne17}, where it implies that the monopolist (on the
demand side) or the monopsonist (on the supply side) absorbs at least half
of any cost increase. This shared \textquotedblleft
distortion\textquotedblright\ property is what makes the two conditions work
together so powerfully. It ensures that low-value buyers---who are already
the most distorted in the aggregate market---are also the ones who stand to
gain the most from segmentation. Conversely, high-value buyers, whose
allocations are already close to efficient, cannot be profitably helped by
rearranging the market.

\subsection{Optimal Hazard Rates and Pricing}

Under these two conditions, the solution to the majorization problem (\ref%
{eq:main-thm}) takes a clean, two-regime form. We recall from the earlier
analysis that $\overline{h}(v)$ denotes the inverse hazard rate that
maximizes local information rents at value $v$ (see (\ref{eq:maxh})). Under
log-concavity of supply, this maximizer is uniquely defined and strictly
increasing in $v$: higher-value buyers can tolerate more distortion before
their rents peak, because their baseline surplus is larger. Meanwhile,
log-concavity of demand makes the aggregate inverse hazard rate $h^{\ast
}(v) $ decreasing in $v$: the aggregate inverse hazard rate falls with $v$,
since higher values are progressively rarer relative to the mass just below
them. These two monotonicity properties guarantee that $h^{\ast }$ and $%
\overline{h}$ cross exactly once. Define the threshold: 
\begin{equation*}
\widehat{v}\equiv \min \big\{v\mid h^{\ast }(v)\leq \overline{h}(v)\big\}
\end{equation*}%
to be the unique threshold value at which the aggregate inverse hazard rate $%
h^{\ast }(v)$ exceeds the rent-maximizing level $\overline{h}(v)$.

\begin{proposition}[Convex Segmentations]
\label{prop:cmv} With log-concave demand and supply, the consumer-optimal
segmentation displays:
\begin{equation}
h(v)=%
\begin{cases}
\overline{h}(v), & \text{if }v<\widehat{v}; \\ 
h^{\ast }(v), & \text{if }v\geq \widehat{v}.%
\end{cases}
\label{fcx}
\end{equation}
\end{proposition}

The logic is intuitive. For buyers above the threshold $\widehat{v}$, the
aggregate inverse hazard rate is already at or below the rent-maximizing
level---there are \textquotedblleft too few\textquotedblright\ high-value
buyers relative to these values for the designer to improve their rents.
Since local information rents are increasing in $h$ at these values, the
designer would like to raise $h$ further, but the majorization constraint is
binding: no feasible reallocation of buyers can push $h$ higher. So these
values are left untouched.

For buyers below $\widehat{v}$, the situation is reversed. The aggregate
market has \textquotedblleft too many\textquotedblright\ higher-value buyers
relative to these low values---the inverse hazard rate $h^{\ast }(v)$ exceeds
the rent-maximizing level $\overline{h}(v)$. This means local information
rents are decreasing in $h$, and the consumer surplus can improve by
concentrating low-value buyers into fewer segments, driving down their
inverse hazard rate to the rent-maximizing level. Crucially, no dilution is
ever needed: the designer uses only the simpler of the two tools identified
earlier, concentrating values rather than creating gaps in segment supports.
This is precisely because log-concave supply makes local information rents
concave in $h$, while the log-concave demand ensures that the marginal local
informational rents are monotonically ordered across values.

The two regularity conditions not only simplify finding the optimal inverse
hazard rates---they also yield an especially clean segmentation structure.

\begin{theorem}[Convex Segmentations and Pricing]
\label{thm:conv-seg} With log-concave demand and supply, there is a
consumer-optimal segmentation such that:

\begin{enumerate}
\item each market $m$ is absolutely continuous with support $[v_{m},%
\overline{v}] $ for some $v_{m}\leq \widehat{v}$, and

\item the pricing differs across markets only by the base price: 
\begin{equation*}
p_{m}(q)=p(q)+T_{m},
\end{equation*}%
where $p(q)$ does not depend on $m$ and $T_{m}$ increases in $v_{m}$ but
does not depend on $q$.
\end{enumerate}
\end{theorem}

Part 1 says that segments are nested convex intervals, all sharing the same
upper bound $\overline{v}$ but with lower bounds that vary across segments.
As concentration increases, fewer low-value buyers remain, and the lowest
values \textquotedblleft drop out\textquotedblright\ first. We refer to these
segmentations as \emph{convex segmentations}, as their supports are nested
convex intervals; they are the natural generalization of the direct
segmentations of \cite{bebm15}.

Part 2 follows from the envelope theorem. By Theorem \ref{thm:main},
every buyer of value $v$ consumes the same quality in every segment in which it
appears. Within each segment,
the indirect utility of value $v$ has the same derivative $q(v)$ at every point in its support. Since the supports of different segments are nested convex intervals, the indirect utilities of any two segments differ only by a constant. The payment of each value is therefore the same up
to a quality-independent base price: the monopolist does not redesign its
product line across segments, but merely adjusts the base price of the lowest
quality offered. What varies is only the level of information rents---because $%
T_{m}$ increases in $v_{m}$, buyers in more concentrated segments (lower $v_{m}$)
face a lower base price and enjoy higher rents.

Alternatively, we can associate each nested segment $[v_{m},\overline{v}]$
with a nested menu as the qualities below $Q(v_{m}-h( v_{m})) $ fail to sell
at the base price $T_{m}$. In this interpretation, the monopolist offers
nested menus all sharing the same efficient upper bound quality $Q( 
\overline{v}) ,$ differing only in how far down the menu is extended and in
the price of the lowest offered quality, the base price.

The theorem shows that the optimal segmentation is easy to construct. Each
segment $m$ has a unique support, which is of the lower censored form $%
[v_{m},\overline{v}]$. Within each segment, the inverse hazard rate of values
below $\widehat{v}$ has been lowered to $\overline{h}$, while the values
above $\widehat{v}$ are unchanged. The key simplifying aspect is optimal
segmentation increasingly concentrates low values more than high values. So,
as low values are concentrated more, we \textquotedblleft run
out\textquotedblright\ of consumers sequentially from $\underline{v}$
upwards, creating segments with support shrinking from the bottom. The
combination of \eqref{eq:mhr} and \eqref{eq:conv-mc} ensures that we run out
of consumers sequentially.

\subsection{The Elasticity Nexus:\ The Case of Iso-Elastic Cost\label%
{subsec:iso-cost}}

To make the interplay between supply and demand elasticities fully
transparent, we now specialize to the family of power cost functions: 
\begin{equation}
c(q)=q^{\gamma }/\gamma ,  \label{eq:iso-cost}
\end{equation}%
where $\gamma >1$. We also refer to the cost as \emph{iso-elastic} as 
\begin{equation*}
\frac{dc\left( q\right) /dq}{c\left( q\right) /q}=\gamma \text{.}
\end{equation*}%
This family of cost functions satisfies log-concavity \eqref{eq:conv-mc} for
all $\gamma $ and nests several important cases: $\gamma =2$ gives quadratic
costs (as in \cite{muro78}), $\gamma \rightarrow \infty $ approximates the
unit demand model, and intermediate values parameterize the curvature of the
cost function. The supply function $Q$ takes the constant elasticity form: 
\begin{equation*}
Q(v)=\mathbbm{1}[v\geq 0]v^{\frac{1}{\gamma -1}}.
\end{equation*}%
with supply elasticity: 
\begin{equation*}
\frac{dQ(v)}{dv}\cdot \frac{v}{Q(v)}=\frac{1}{\gamma -1}\in (0,\infty ).
\end{equation*}%
A direct calculation shows that with the power cost function the inverse
hazard rate $\overline{h}$ that maximizes the local information rent has
the explicit form: 
\begin{equation}
\overline{h}(v)=\frac{\gamma -1}{\gamma }v.  \label{eq:gamma}
\end{equation}%
Now, the only class of distributions that has a linear inverse hazard rate
is the class of Pareto distributions given by 
\begin{equation}
F_{P}\left( v\right) \equiv 1-\left( \underline{v}/v\right) ^{\alpha },
\label{eq:paret}
\end{equation}%
where $\alpha \in \mathbb{R}_{+}$ is referred to as the shape parameter of
the Pareto distribution. The inverse hazard rate of the Pareto distribution
with shape parameter $\alpha$ is: 
\begin{equation*}
\frac{1-F_{P}\left( v\right) }{f_{P}\left( v\right) }=\frac{v}{\alpha }\text{%
.}
\end{equation*}%
An immediate Corollary of Proposition \ref{prop:cmv} for the case of
iso-elastic cost functions follows.

\begin{corollary}[Segmentation with Pareto Distribution]
\label{cor:cmvc}With iso-elastic cost and log-concave demand, there is a
consumer-optimal segmentation where each segment $m$ is composed of a Pareto
distribution and the aggregate demand: 
\begin{equation}
h_{m}(v)=%
\begin{cases}
\frac{\gamma -1}{\gamma }v, & \text{if }v<\widehat{v}; \\ 
h^{\ast }(v), & \text{if }v\geq \widehat{v};%
\end{cases}
\label{eq:cmvc1}
\end{equation}%
where the shape parameter of the Pareto distribution is $\alpha =\gamma
/(\gamma -1)$.
\end{corollary}

Thus, we have obtained an explicit shape for every segment. It is the
pairing of the Pareto distribution on the lower end and the aggregate
distribution itself on the upper end of the distribution. The shape
(parameter $\alpha $) of the Pareto distribution is determined by the
convexity (exponent $\gamma $) of the cost function.

We can restate the above Corollary in terms of the demand elasticity $\eta
_{m}(v)$ in every segment $m$. The demand elasticity is given by: 
\begin{equation*}
\eta _{m}(v)\equiv -\frac{f_{m}(v)v}{1-F_{m}(v)}=-\frac{v}{h_{m}(v)}.
\end{equation*}

\begin{corollary}[Segmentation and Demand Elasticity]
\label{cor:cmve}With iso-elastic cost and log-concave demand, there is a
consumer-optimal segmentation where each segment $m$ has constant elasticity
below a cutoff and matches the aggregate market elasticity above it:%
\begin{equation*}
\eta _{m}(v)=%
\begin{cases}
\frac{\gamma }{1-\gamma }, & \text{if }v<\widehat{v}; \\ 
\eta ^{\ast }(v), & \text{if }v\geq \widehat{v}.%
\end{cases}%
\end{equation*}
\end{corollary}

The corollary delivers a sharp economic message. The consumer-optimal
segmentation makes demand \textquotedblleft elastic
enough\textquotedblright\ at the bottom of the distribution, where the
monopolist's distortions are most severe. Above the threshold $\widehat{v}$,
demand elasticities are left as they are. Below $\widehat{v}$, they are set
to a constant that depends only on the cost structure---specifically, $%
\gamma /(1-\gamma )$, which is one plus the supply elasticity, with a sign
reversal.

This connection between supply and demand elasticities has a natural
interpretation. More elastic supply (lower $\gamma$) means the monopolist
can adjust quality more cheaply, amplifying the distortions it imposes on
low-value buyers. To counteract these larger distortions, the optimal
segmentation must make demand correspondingly more elastic at low values,
raising the \textquotedblleft cost\textquotedblright\ to the monopolist of
excluding or under-serving these buyers. More rigid supply (higher $\gamma$)
means the monopolist's quality choices are less responsive to market
composition, so less demand-side adjustment is needed.

\subsection{When Does Segmentation Always Harm Consumers?}

Perhaps the most policy-relevant question our framework can answer is: when
should we expect segmentation to hurt consumers regardless of how the market
is carved up? An immediate consequence of Proposition \ref{prop:cmv}
provides a sharp answer.

\begin{corollary}[Zero Segmentation]
\label{cor:no-seg2} Under log-concave demand and supply, zero segmentation
is optimal if and only if 
\begin{equation}
h^{\ast }({\underline{v}})\leq \overline{h}({\underline{v}}).
\label{zeroseg}
\end{equation}
\end{corollary}

The condition is appealingly simple: segmentation cannot help consumers when
the inverse hazard rate at the \emph{lowest} value ${\underline{v}}$ is
already below the rent-maximizing level. Since \eqref{eq:mhr} makes $h^{\ast
}$ decreasing, if the condition holds at ${\underline{v}}$, it holds
everywhere, so the designer cannot improve rents at any value. The economic
interpretation is that when demand is sufficiently elastic throughout the
market---when there are relatively few high-value buyers per unit of each
lower value---the monopolist's screening already provides consumers with
near-optimal informational rents. Any redistribution of values across
segments can only reduce them.

When the cost is iso-elastic, the zero-segmentation condition \eqref{zeroseg}
can be written as follows: 
\begin{equation}
\eta ^{\ast }(\underline{v})\leq \frac{\gamma }{1-\gamma },  \label{eq:ela}
\end{equation}%
which crystallizes the interplay between supply and demand elasticities. In
particular, when demand is given by a parametrized class of distributions,
we can use the condition (\ref{eq:ela}) to directly verify the optimality of
zero segmentation. For example, suppose the demand is given by the above
class of Pareto distributions. The demand elasticity is then constant and
given by the negative of the shape parameter, $-\alpha $. From Corollary \ref%
{cor:no-seg2} we then have that zero segmentation is optimal when 
\begin{equation*}
-\alpha \leq \frac{\gamma }{1-\gamma }\iff \alpha \geq \frac{\gamma }{\gamma
-1},
\end{equation*}
and thus demand is sufficiently elastic relative to supply. We can derive
similar condition for the class of exponential distributions or the class of
uniform distributions, both of which have monotone demand elasticities.

Corollary \ref{cor:no-seg2} stands in contrast with \cite{hasi23}: they show
that for generic markets, there is always a segmentation that improves
consumer surplus. The conditions for Corollary \ref{cor:no-seg2} to go
through are not too restrictive and, in particular, satisfy the notion of
genericity in \cite{hasi23}. The discrepancy arises because they consider
environments with a discrete number of goods, which does not satisfy %
\eqref{eq:conv-mc} (in fact, they provide a counterexample showing that
their result fails with a continuum of goods). Corollary \ref{cor:no-seg2}
tells us that as the number of discrete goods grows large, the gains from
segmentation they identify sometimes (but not always) disappear, and
provides a characterization of when this happens.

Using Theorem \ref{thm:main}, we can extend Corollary \ref{cor:no-seg2} to
the general environment. The general characterization of when zero
segmentation is optimal is somewhat technical, so we relegate it to the
appendix (see Corollary \ref{cor:no-seg}). Instead, further focusing the
cost function to be iso-elastic we can gain further insight into how cost
and demand determine when zero segmentation is optimal.

When cost is more elastic (lower $\gamma $), the threshold $\gamma /(
1-\gamma )$ becomes more negative, making the condition harder to
satisfy: the set of markets where zero segmentation is optimal \emph{shrinks}%
. A more elastic supply means the seller can adjust quality more readily,
amplifying the distortions that segmentation can exploit. Conversely, a more
inelastic supply (higher $\gamma $) makes quality provision rigid, reducing
the scope for segmentation to improve consumer welfare. Hence, the condition
makes it clear that the set of markets where zero segmentation is optimal
grows larger as the exponent $\gamma $ increases. In fact, this is true even
dropping the \eqref{eq:mhr} assumption. Denote by $Z_{\gamma }$ the set of
markets under which zero segmentation is optimal when the cost function is
given by \eqref{eq:iso-cost}.

\begin{proposition}[Zero Segmentation---Iso-elastic Cost]
\label{prop:com} With iso-elastic cost, for any $\gamma ^{\prime }<\gamma $, 
$Z_{\gamma ^{\prime }}\subset Z_{\gamma }$.
\end{proposition}

Hence, as the cost becomes more elastic ($\gamma $ decreases), the potential
gains from segmentation increase. This is despite the fact that even in some
aggregate markets $m^{\ast }$ associated with inefficient allocations, no
segmentation can be optimal. In contrast, in the limit $\gamma \rightarrow
\infty $, zero segmentation is optimal only if the allocation in the
aggregate market $m^{\ast }$ itself is efficient. One might have thought
that this means that it is relatively unlikely to find a market in which
zero segmentation is optimal. However, the conclusion is the opposite: in
the limit $\gamma \rightarrow \infty $, the cost is very inelastic, reducing
the potential benefits from segmentation.

\section{Generalization and Applications\label{sec:gen-exts}}

The analysis of Sections \ref{sec:opt-seg} and \ref{sec:conv-seg} was
developed for a specific objective (consumer surplus maximization) in a
specific environment (private values with a convex cost of quality). In this
section, we show that the underlying framework---the reduction to an
optimization over inverse hazard rates via the local information
rent---extends well beyond this setting. The key insight of Theorem \ref%
{thm:main}, reducing the segmentation problem to an optimization over
inverse hazard rates via the local information rent, carries over to
settings with adverse selection, alternative welfare objectives, and richer
preference structures.

The crucial property of the local informational rents that we used in our
analysis is: 
\begin{equation}
\frac{\partial ^{2}u(v,h)}{\partial h^{2}}+\frac{\partial ^{2}u(v,h)}{%
\partial h\partial v}<0,  \label{eq:smcu}
\end{equation}%
for all $(v,h)$ with $v>h$, thus the marginal local information rent
$\partial u/\partial h$ decreases along the diagonal direction in $(v,h)$.
We used this property to show that the solution
to \eqref{eq:ub} is regular. One can heuristically verify this by writing
the derivative of the optimality condition (\ref{eq:d-gaph}) explicitly; it
looks quite similar to \eqref{eq:smcu}, the only difference being that the
concavity of $u$ is multiplied by $h^\prime(v)$ in the optimality condition.
Now, the concavity is negative, so it must be multiplied by a number smaller
than one for the optimality condition (\ref{eq:d-gaph}) to have a different
sign than \eqref{eq:smcu}. Hence, when maximizing over inverse hazard rates %
\eqref{eq:ub} we get that a regular inverse hazard rate (ie., $%
h^\prime(v)\leq1$) is necessary for optimality.

A generalization of Theorem \ref{thm:main}, stated and established as
Theorem \ref{thm:gen} in the appendix, will apply to any objective function
that satisfies \eqref{eq:smcu}. We now discuss two extensions: adverse
selection environments, where the seller's cost depends on the buyer's type;
and Pareto-efficient segmentations, where the designer balances consumer and
total surplus. Finally we briefly discuss other extensions where the main
elements of Theorem \ref{thm:main} apply such as finite value or weakly
convex costs.

\subsection{Adverse Selection and Market Fragmentation}

Our baseline model assumes that the cost of quality provision is independent
of the buyer's type---a private-values environment. Many important
applications, however, feature adverse selection: the seller's cost rises
with the buyer's willingness to pay. Insurance markets, credit markets, and
financial trading all share this feature, and it fundamentally changes the
economics of screening. To incorporate adverse selection while maintaining a
transparent connection to our baseline, we consider the following
specification. A buyer who trades quantity $q$ obtains utility:%
\begin{equation}
\tau \left( v\right) q-p-\frac{q^{2}}{2},  \label{adv:1}
\end{equation}%
where $\tau (v)$ is the buyer's gross value and $v$ is the net
value known by the buyer, distributed according to $F^{\ast }(v)$. Now $q$
is interpreted as a quantity rather than a quality. The quadratic term
captures diminishing marginal utility, as in \cite{mari84}, or risk aversion
in asset trading models as in \cite{bimr00}. The cost of supplying the good
is: 
\begin{equation}
c(q)=(\tau (v)-v)q,  \label{adv:2}
\end{equation}%
which is linear in quantity but type-dependent. The derivative of $\tau $
measures the severity of adverse selection, and we assume that: 
\begin{equation}
\frac{d\tau (v)}{dv}\in \lbrack 1,\infty ),  \label{adv:3}
\end{equation}%
which implies that the cost is (weakly) increasing in the value $v$. As a
special case, if the derivative is equal to 1 at every $v$, the marginal
cost is independent of $v$, as in a private value environment. Otherwise, we
recover the trading model with adverse selection of \cite{bimr00}. We
additionally impose that: 
\begin{equation}
0\leq \frac{\tau ^{\prime \prime }(v)v}{\tau ^{\prime }(v)}\leq 1,
\label{adv:4}
\end{equation}%
which requires the degree of adverse selection to be increasing in $v$ but
not too quickly.

The local information rents in this environment take the form: 
\begin{equation}
w(v,h)\equiv u\left(v,\frac{d\tau (v)}{dv}h\right),  \label{adv:5}
\end{equation}%
where now the supply is $Q(v)=v$ because the utility is quadratic in $q$.
Adverse selection enters by amplifying the inverse hazard rate: the
effective inverse hazard rate facing the designer is $\tau ^{\prime }(v)h$
rather than $h$ alone. This amplification has a clear economic
interpretation. Under adverse selection, the seller's screening problem is
more severe---each unit of information asymmetry (as captured by $h$) does
more damage---because the seller must account for both the downward
distortion of quantities and the type-dependent cost wedge.

\begin{proposition}[Optimal Segmentation with Adverse Selection]
\label{ext:adv} In the model of adverse selection \eqref{adv:1}-\eqref{adv:4}%
, the maximum consumer surplus attained by segmentation is: 
\begin{equation}
\max_{\sigma \in \operatorname{MPS}(m^{\ast })}\left[ \int_{\Delta V}U(m)d\sigma (m)%
\right] =\max_{h\prec h^{\ast }}\ \int_{\underline{v}}^{\overline{v}%
}w(v,h(v))\ dF^{\ast }(v),  \label{thm:advs}
\end{equation}%
where $w$ is given by \eqref{adv:5}. Furthermore, every segmentation
achieving this value is a uniform segmentation implementing some $h$ which
solves the rhs of \eqref{thm:advs}
\end{proposition}

The result shows that the analytical framework---the reduction to
inverse hazard rates, the majorization constraint, and the optimality of
uniform segmentations---carries over to this environment. What changes is the
shape of the objective function, and with it the quantitative features of the
optimal segmentation. In particular, greater adverse selection (higher $\tau
^{\prime }$) leads the optimal segmentation to feature more elastic demands
within each segment. The intuition is that adverse selection already
depresses trade; the designer responds by making segments more elastic, so
that more generous terms of trade can be sustained.

It is worth noting why the reduction goes through under
\eqref{adv:1}-\eqref{adv:4}, and where it would not. Adverse selection enters
this specification only by rescaling the inverse hazard rate, $h\mapsto
\tau ^{\prime }(v)h$, leaving an objective of the same form $w(v,h)=u(v,\tau
^{\prime }(v)h)$. Three assumptions do the work. The linear-quadratic payoff in
\eqref{adv:1} delivers a linear supply $Q(v)=v$, so the local rent retains the
product structure $h \cdot Q(\cdot)$ underlying Lemma \ref{lm:csr}. Condition
\eqref{adv:3} keeps $v$ an increasing index of willingness to pay, so segment
regularity is still governed by monotone virtual values. Most importantly,
\eqref{adv:4} bounds how fast the severity of adverse selection grows in $v$,
which is exactly what preserves the inequality \eqref{eq:smcu} driving the
optimality of regular, uniform segmentations. The reduction does \emph{not}
extend automatically beyond these conditions: if \eqref{adv:4} fails, so that
adverse selection intensifies too quickly in $v$, then $w$ can violate
\eqref{eq:smcu} and the optimality of uniform segmentations may break down.

\paragraph{Market Fragmentation}

In this context, segmentation has a natural reinterpretation as market
fragmentation: buyers trade across different venues, each attracting a
different mix of types. A monopolist market maker sets terms in each venue,
and our analysis characterizes the fragmentation structure that maximizes
buyer welfare. This connection to market microstructure opens a broader
reading of our results. The question ``when does market
segmentation help consumers?'' maps onto ``when does fragmenting a trading venue improve the terms of
trade?'' Our characterization---the answer hinges on the
interaction of value-dependent costs and the distribution of buyer values---
provides a foundation for analyzing fragmentation in financial markets. The
assumption that the good is supplied by a single seller (market maker)
provides a natural benchmark (see, for example, \cite{glos89}), however,
extending the model to allow buyers to submit orders to multiple venues
simultaneously, as in \cite{maro17}, is a natural and promising direction.

\subsection{Pareto-Efficient Segmentations}

Our main analysis focused on maximizing consumer surplus, but a regulator or
a platform may care about broader welfare objectives. We now consider the
problem of finding segmentations that are Pareto efficient---those that
maximize a weighted sum of consumer surplus $U(m)$ and social surplus $S(m)$:
\begin{equation*}
U(m)+\lambda S(m),
\end{equation*}%
where $\lambda \geq 0$ represents the relative weight on total welfare (or,
equivalently, on the seller's profits, since social surplus includes both
consumer surplus and profits). Setting $\lambda =0$ recovers our baseline
problem; the limit $\lambda \rightarrow \infty $ corresponds to maximizing
social surplus alone. Any objective placing a larger weight on profits than
on consumer surplus yields the same solution as this limit, so varying $%
\lambda $ traces the entire Pareto frontier. We apply Theorem \ref{thm:gen}
by setting: 
\begin{equation}
w(v,h)\equiv u(v,h)+\lambda (vQ(v-h)-c(Q(v-h))).  \label{w2}
\end{equation}%
Here $u(v,h)$ gives the local informational rents while the second term
multiplying $\lambda $ is the local social surplus generated by value $v$.

\begin{proposition}[Linear Combination of Consumer and Social Surplus]
\label{ext:sw} If $\lambda\geq-1$, the maximum linear combination of social
surplus and consumer surplus attained by segmentation is: 
\begin{equation}  \label{thm:PE}
\max_{\sigma \in \operatorname{MPS}(m^\ast)} \left[\int_{\Delta V} U(m)+\lambda
S(m)\ d\sigma(m) \right] = \max_{h\prec h^\ast}\ \int_{\underline
v}^{\overline v}w(v,h(v))\ dF^\ast(v),
\end{equation}
where $w$ is given by \eqref{w2}. Furthermore, every segmentation achieving
this value is a uniform segmentation implementing some $h$ which solves the
rhs of \eqref{thm:PE}
\end{proposition}

The proposition confirms that the majorization framework applies across the
Pareto frontier. What changes as $\lambda $ increases is the shape of $w$
and, consequently, the optimal inverse hazard rates. Higher $\lambda 
$ pushes the optimal segmentation toward more elastic demands within each
segment: reducing distortions becomes more valuable when social surplus
carries greater weight. In the limit $\lambda \rightarrow \infty $, demands
become perfectly elastic---the efficient outcome---which corresponds to
first-degree price discrimination. This makes economic sense: when the
objective is purely allocative efficiency, the designer's ideal is to
eliminate all screening distortions, which requires each segment to contain
a single value. Under iso-elastic costs, as introduced in Section \ref%
{sec:conv-seg}, the demand elasticity in each segment below the threshold
takes the closed-form value $(\gamma+\lambda )/( 1-\gamma)$. As $\lambda $
increases, this elasticity rises monotonically, confirming the intuition
that moving along the Pareto frontier toward greater total surplus involves
progressively reducing market power within segments. Beyond
the consumer-optimal level, further increases in demand elasticity no longer
benefit consumers: they translate directly into gains for the seller.

Figure \ref{fig:pareto} illustrates the Pareto frontier for a uniform
aggregate market $F^{\ast }=\mathcal{U}[0,1]$ and quadratic cost $c(q)=q^{2}/2$. As $\lambda $
ranges over $[-1,\infty )$, the induced pairs $(U,\Pi )$ trace the rightward boundary of
the set attainable by segmentation. Two features stand out. First, the
consumer-optimal segmentation ($\lambda =0$) delivers strictly more consumer
surplus \emph{and} profit than no segmentation:
here, some segmentation benefits buyers and seller alike. Second, as $\lambda $ rises
beyond the consumer-optimal level, the optimum moves along the frontier, trading
consumer surplus for profit, and converges as $\lambda \rightarrow \infty $ to
the efficient outcome implemented by first-degree price discrimination. The outcome is socially efficient only at that point;
elsewhere total surplus falls short because the within-segment
screening distortion cannot be removed while leaving rents to consumers. Unlike the unit-demand benchmark of \cite{bebm15}, the full
surplus triangle is not attainable under screening within segments (cf.\
\cite{hasi22}).

\begin{figure}[th]
\centering
\includegraphics[width=0.55\textwidth]{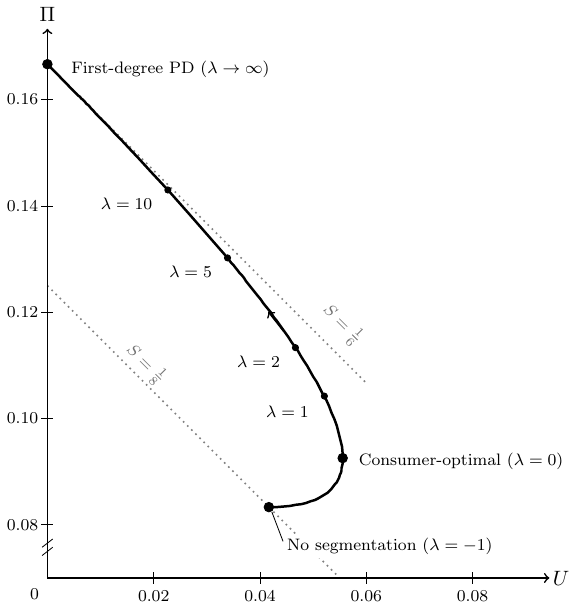}
\caption{Pareto frontier of $(U,\Pi)$ achievable with segmentation, where 
$c(q)=q^{2}/2$ and $F^{\ast }=\mathcal{U}[0,1]$. Dotted lines are iso-surplus loci
$U+\Pi =S$.}
\label{fig:pareto}
\end{figure}

\subsection{Other Extensions\label{subsec:other-ext}}

\paragraph{Finite Values}

When the aggregate value distribution is discrete, the analysis
carries through with natural modifications. The objective must be replaced
by its concave envelope in the inverse hazard rate, and the notion of
uniform segmentation is relaxed to an average implementation condition.
Buyers of the same value may receive different allocations across segments,
but allocations remain monotone in value. As the discrete model approximates
a continuum, these differences vanish, and outcomes converge to the
continuous benchmark.

\paragraph{Weakly Convex Cost}

The uniqueness of the optimal solution (see discussion after Theorem \ref%
{thm:main}) crucially relies on the strict convexity of cost; if cost is
only weakly convex, the solution we characterize continues to be optimal,
but non-uniform segmentations can also be optimal. For example, in the
unit-demand environment, there are optimal segmentations with varying
inverse hazard rates across markets and also with irregular markets. We thus
obtain much stronger predictions by assuming the cost function is strictly
convex rather than weakly convex.

\section{Conclusion\label{sec:con}}

This paper characterizes how market segmentation affects consumer welfare
when monopolists can engage in both second- and third-degree price
discrimination. Our analysis yields several key insights. First, the
consumer-optimal segmentation maintains consistent quality provision across
segments while allowing price variation. Second, the benefits of
segmentation depend critically on demand and cost elasticities, with no
segmentation being optimal when aggregate demand is sufficiently elastic.

Our central methodological contribution is to show that the consumer-optimal
segmentation, despite being defined over an infinite-dimensional space of
distributions of distributions, reduces to maximizing the expected local
information rent over a single inverse hazard rate function subject to a
majorization constraint. This reduction enables a complete characterization:
consumers of the same value receive the same quality in every segment; the
optimal segmentation can be implemented through a common quality menu with
segment-specific fixed fees; and whether segmentation helps consumers at all
is determined by a sharp threshold on the interaction of demand elasticities
and cost convexity.

These theoretical results have direct practical implications. For
competition authorities, they suggest that market segmentation should be
evaluated based on observable market characteristics. When consumer demand
is sufficiently elastic across the whole market, segmentation is uniformly
bad for consumers. The monopolist's screening already provides information
rents close to what any information structure could achieve, and
segmentation can only dilute these rents by redistributing buyer values
inefficiently. This contrasts sharply with the unit-demand benchmark, where
segmentation is always potentially beneficial, and it provides a concrete,
testable criterion for when segmentation raises consumer welfare. For firms,
our characterization of optimal segmentation strategies offers guidance for
designing market segmentation policies that balance profit maximization with
consumer welfare. Our analysis also connects to broader debates about big
data and personalized pricing in digital markets. While enhanced ability to
segment markets could enable more sophisticated price discrimination, our
results suggest this may benefit consumers when properly structured.

Several important directions remain for future research. First, extending
the analysis to competitive markets as in the recent analysis of \cite%
{bebm25} for single unit demand would provide insight into how market
structure affects optimal segmentation. Second, empirical work testing our
theoretical predictions about the relationship between demand elasticities
and optimal segmentation would be valuable.

\newpage

\appendix

\renewcommand{\thelemma}{\Alph{section}.\arabic{lemma}} \renewcommand{%
\theproposition}{\Alph{section}.\arabic{proposition}} \renewcommand{%
\thecorollary}{\Alph{section}.\arabic{corollary}} \setcounter{lemma}{0} %
\setcounter{proposition}{0} \setcounter{corollary}{0} %
\numberwithin{figure}{section} \numberwithin{equation}{section}

\section{Appendix}

\label{sec:proofs}

\begin{proof}[Proof of Lemma \ref{lm:csir}]
We begin by describing the profit-maximizing menu. Define
\[
J_m(t) \equiv F^{-1}_m(t)(t-1).
\] 
We denote by $t_m$ the largest minimizer of $J_m$ and by $\mathrm{vex}[J_m]$ the convexification (i.e.\ the pointwise largest convex function smaller than $J_m(t)$). Then, a buyer with value $v$ buys quality $q$ satisfying:
    \[c^\prime(q)=\frac{\partial}{\partial t}\mathrm{vex}\left[J_{m}\right](F_m(v)),\]
if $F_m(v)\geq t_m$, and $q=0$ otherwise. Note that $J^{\prime}_m(F_m(v))=\phi_m(v)$, so the expression coincides with the one in the main text when $m$ is regular.

Consider the set of markets $X_{p_m}$ where menu $p_m$ is seller-optimal:
\begin{equation*}
    X_{p_m} \equiv \big\{m^{\prime} \in \Delta V: \Pi(m^{\prime},p_{m^{\prime}}) = \Pi(m^{\prime},p_m)\big\}.
\end{equation*}
The seller's profit $\Pi(m',p_{m'})$ is a pointwise supremum of the weak$^\ast$-continuous affine maps $m'\mapsto\Pi(m',p)$, and hence convex and weak$^\ast$-lower semicontinuous (lsc); since $\Pi(m',p_{m'})\geq\Pi(m',p_m)$ and $m'\mapsto\Pi(m',p_m)$ is affine and weak$^\ast$-continuous, $X_{p_m}$ is a sublevel set of a convex weak$^\ast$-lsc function, and therefore convex and weak$^\ast$-closed (Theorem 7.6 and Lemma 2.42, \cite{albo06}). Because $V$ is a compact metric space, $\Delta V$ is itself a convex, weak$^\ast$-compact, metrizable set (Theorem 15.11, \cite{albo06}), and the closed, convex subset $X_{p_m}$ inherits all three properties. Therefore, by Choquet's theorem (\cite{phel01}), $m \in X_{p_m}$ can be written as a segmentation supported on its extreme points.

Now, observe that $m^{\prime}\in X_{p_m}$ if and only if
\[
t_m = t_{m^\prime} \text{ and } \mathrm{vex}\left[J_{m}\right](t)=\mathrm{vex}\left[J_{m^{\prime}}\right](t) \text{ for all } t \geq t_m.
\]
On the other hand, $m$ is irregular if and only if $J_m(t)>\mathrm{vex}[J_m](t)$ on some interval $[t_1,t_2]$ where $F_m^{-1}(t)$ is not constant. For such $m$, sufficiently small perturbations of $F_m(v)$ over  $v \in [F_m^{-1}(t_1), F_m^{-1}(t_2)]$ leave $\mathrm{vex}[J_m]$ unchanged (and hence remain in $X_{p_m}$), meaning irregular markets are not extreme points. Thus, irregular markets can be segmented into regular markets without changing the profit-maximizing menu.

The preceding reduction also allows us to restrict to markets without singular-continuous components. Let $A_m \subseteq (0,1)$ denote the union of the interiors of all maximal intervals on which $F_m^{-1}$ is constant; these correspond exactly to the atoms of $F_m$. Define
\[
 H(t) \equiv \frac{\operatorname{vex}[J_m](t)}{t-1} \text{ for } t \in (0,1).
\]
Since $\operatorname{vex}[J_m]$ is convex, $H$ is locally absolutely continuous. Moreover, regularity implies $F_m^{-1}(t) = H(t)$ for Lebesgue-a.e.\ $t \in A_m^c$. Take any $t \in A_m^c$ where $H'(t)$ exists; except at the endpoint of a terminal constant interval, there exists $s > t$ such that $F_m^{-1}(s) > F_m^{-1}(t)$. Convexity and the fact that $\operatorname{vex}[J_m] \leq J_m$ implies
\[
\operatorname{vex}[J_m]'(t) \leq \frac{\operatorname{vex}[J_m](s)-\operatorname{vex}[J_m](t)}{s-t} \leq \frac{J_m(s)-J_m(t)}{s-t}.
\]
Thus,
\begin{equation}\label{eq:iron-q-d}
    {H}'(t) \geq \frac{1-s}{1-t} \cdot \frac{F_m^{-1}(s)-F_m^{-1}(t)}{s-t} > 0 \text{ for Lebesgue-a.e.\ } t \in A_m^c.
\end{equation}

Let $B \subseteq V$ be any set with Lebesgue measure zero containing no values carrying atoms, meaning that $F_m^{-1}(t) \in B \implies t \in A_m^c$. Since $H$ is locally absolutely continuous, $H'$ exists a.e. on $(0,1)$, and the Serrin-Varberg theorem (\cite{seva69}) gives
\[
\lambda(\{t : H'(t) \neq 0 \text{ and } H(t) \in B\}) = 0,
\]
where $\lambda$ denotes the Lebesgue measure. Combining this with \eqref{eq:iron-q-d} implies that
\[
\lambda(\{t \in A_m^c: H(t) \in B\}) = 0.
\]
Finally, since $F_m^{-1}(t) = H(t)$ for Lebesgue-a.e.\ $t \in A_m^c$,
\[
m(B)=\lambda(\{t : F_m^{-1}(t)\in B\}) \leq \lambda(\{t \in A_m^c : H(t)\in B\}) = 0,
\]
meaning $m$ has no singular-continuous component.
\end{proof}

To make the notation more compact, we denote the survival function by: 
\begin{equation*}
D_{m}\equiv 1-F_{m}.
\end{equation*}

\begin{proof}[Proof of Proposition \ref{prop:avg-h}]
    To begin, we claim that for any regular market $m$,
    \begin{equation}\label{eq:d-gap}
        \int^{\overline v}_v D_m(t)\ dt = 
        \begin{cases}
            \int^{\overline v}_vh_m(t)\ dF_m(t) & \text{if } \Delta_m(v)=0; \\
            \int^{\overline v}_vh_m(t)\ dF_m(t) + \Delta_m(v) D_m(v) & \text{if }\Delta_m(v)>0.
    \end{cases}
    \end{equation}
In any regular market, $\Delta_m(v)=0$ if and only if $F_m$ is absolutely continuous at $v$ in the local sense that the continuous non-atomic component admits a density there; by the observation in Lemma \ref{lm:csir}, regularity rules out a singular-continuous component.  If $F_m$ is absolutely continuous at $v$ with density $f_m(v)$, then by definition $D_m(v) = f_m(v) h_m(v)$. Also, in any regular market, at the lower boundary of a gap there is always an atom. That is, at any interval $(v_1,v_2)$ satisfying 
\[(v_1,v_2)\cap\mathrm{supp}(m)=\{\emptyset\}\text{ and }v_1,v_2\in\mathrm{supp}(m),\] we must have that $a(v_1)>0$.    In any such interval,
    \[
    \int_v^{v_2} D_m(t)\ dt = \Delta_m(v)D_m(v)=\Delta_m(v)D_m(v)+  \int^{ v_2}_vh_m(t)\ dF_m(t).
    \]
 where in the second equality we are integrating zero because $dF(v)$ is zero when $v \notin \supp(m)$. At the limit $v=v_1$ we obtain
      \[
    \int_{v_1}^{v_2} D_m(t)\ dt =  a_m(v_1)h_m(v_1) =\lim_{w\uparrow v_1}  \int_w^{v_2} h_m(t)\ dF_m(t)
    \]
where we take the left limit of the lower boundary of the integral to include the atom at $v_1$. Thus, combining the intervals, we get \eqref{eq:d-gap}.

    We now apply this to the aggregate market:
    \begin{multline}\label{eq:eh-proof}
        \int_v^{\overline{v}} h^\ast(t)\ dF^\ast(t) = \int_v^{\overline{v}} \int_{\Delta V}  D_m(t) \ d\sigma(m) \ dt = \int_{\Delta V}  \int_v^{\overline{v}} D_m(t)\ dt \ d\sigma(m) \\
        = \int_{\Delta V}  \int_v^{\overline{v}}h_m(t)\ dF_m(t) \ d\sigma(m)+ {\int_{\Delta V}}\mathbbm{1}\{v\not\in \mathrm{supp}(m)\} \Delta_m(v)D_m(v) \ d\sigma(m),
    \end{multline}
    where the change in the order of integration uses Fubini's theorem (Theorem 11.27, \cite{albo06}) and the second line is applying \eqref{eq:d-gap}. The first term can be written as:
    \[
    \int_{\Delta V} \int_v^{\overline{v}} h_m(t)\ dF_m(t)\ d\sigma(m) = \int_v^{\overline v} \int_{\Delta V} h_m(t)\ d\Sigma_t(m) \ dF^\ast(t) = \int_v^{\overline v} h_\sigma(t)\ dF^\ast(t).
    \]
    This again uses Fubini's theorem and the disintegration/conditional measure notation embedded in the definition of $\Sigma_t$.
    Subtracting this from the lhs of \eqref{eq:eh-proof} yields the result.
\end{proof}

\begin{proof}[Proof of Proposition \ref{prop:unif-seg-ub}]
Endow $\Delta V$ and $\Delta(\Delta V)$ with the weak topology on measures, so that $\Delta V$ and $\Delta(\Delta V)$ are closed and compact. The maximization problem \eqref{eq:ub} is taken over $h \in L_+^p[\underline{v},\overline{v}]$, where $1 < p < \infty$. Boundedness guarantees $Q$ (and $u$) lie in $L_+^p$, and our tie-breaking assumptions ensure that $Q$ and $u$ are upper semicontinuous (usc). The assumption that $f^\ast(v) > 0$ and $F^\ast$ is real-analytic means that $h^\ast$ is real-analytic. A crucial tool for our analysis will be the concave upper envelope of $u(v,h)$ with respect to $h$ (i.e.\ the pointwise smallest concave function larger than $u(v,h)$), which we denote by $\overline{u}(v,h)$. Since $u(v,h)$ is bounded and zero when $h \geq v$, $\overline{u}(v,h)$ is constant for $h \geq \overline{h}(v)$.

The proof is divided into two main parts. Proposition \ref{prop:sup-max} proves that there exists an $h$ which attains the supremum of \eqref{eq:ub}, and characterizes the solution using a KKT system, which allows us to derive certain properties of the solution (Lemma \ref{lm:h-opt-reg}). Proposition \ref{prop:canon-seg} then provides a particular construction method for a uniform segmentation implementing any $h$ which satisfies the technical properties of Lemma \ref{lm:h-opt-reg}.

\begin{proposition}[Supremum is Attained]\label{prop:sup-max}
    The supremum of \eqref{eq:ub} is attained. Furthermore, $h$ solves \eqref{eq:ub} if and only if both \textup{(i)} $\frac{\partial}{\partial h} u(v,h(v))$ is non-decreasing in $v$ and constant in $v$ when $E_h(v) > 0$, and \textup{(ii)} $u(v,h(v)) = \overline{u}(v,h(v))$
    hold at almost all $v$.
\end{proposition}

\begin{proof}
  For any $h(v)$ we have that
    \[\widehat h(v)=\min\{h(v),\overline h(v)\}\]
    generates higher local informational rents and relaxes the majorization constraint. Hence, it is without loss to make the feasible set
    \[
    \mathcal{H} \equiv \Big\{h \in L_+^p[\underline{v},\overline{v}] \mid h \prec h^\ast \text{ and } 0 \leq h(v) \leq \overline{h}(v) \text{ a.e.}\Big\}.
    \]
    We then study the relaxed problem:
    \begin{equation}\label{eq:relaxed}
    \sup_{\mathcal{H}} \int_{\underline{v}}^{\overline{v}} \overline{u}(v,h(v))\ dF^\ast(v).
    \end{equation}
    Endow this space with the weak topology. By construction, $\overline{u}(v,h)$ is concave (and hence also continuous) in $h$, and additionally bounded because $u$ is bounded. By \cite{rock68}'s theorem on integral functionals, the rhs is concave and weakly usc over $\mathcal{H}$. Furthermore, since $\mathcal{H}$ is a closed subset of a weakly compact space, it is itself weakly compact (Theorem 6.25, \cite{albo06}). Hence, by Weierstrass Theorem, the supremum is attained.

    Next, we wish to characterize the solution of \eqref{eq:relaxed} using a KKT system. The classic Slater constraint qualification fails, due to the fact that both the box and majorization constraints have no topological interior in their corresponding function spaces. However, by applying an appropriate generalization of this constraint qualification to infinite dimensional settings (\cite{jewo92,bocw08}), we get that the KKT conditions are both necessary and sufficient.\footnote{In particular, we need a feasible point in the \emph{quasi-relative interior} of $\mathcal{H}$, i.e.\ an $H \in \mathcal{H}$ such that
    \[
    0 < h(v) < \overline{h}(v) \text{ a.e.\ and} \int_v^{\overline v} h(t)\ dF^\ast(t) < \int_v^{\overline v} h^\ast(t)\ dF^\ast(t) \text{ for all } v < \overline v,
    \]
    which is satisfied by truncating $h^\ast$ by $\overline{h}$ and then scaling down by $0 < \alpha < 1$.}
    Thus, for any optimal $h \prec h^\ast$, there exists a function $\lambda(v) \in L_+^q[\underline{v},\overline{v}]$, where $\frac{1}{p} + \frac{1}{q} = 1$, and a finite non-negative measure $\mu$ such that (Theorem 13.18, \cite{albo06})
    \[
    \frac{\partial}{\partial h}\overline{u}(v,h(v)) = \mu[\underline v,v] - \lambda(v) \text{ and } \lambda(v) h(v) = E_h(v) d\mu(v) = 0 \text{ a.e.}
    \]
    Here $\lambda$ is the multiplier on the non-negativity constraint, and $\mu$ the multiplier on the majorization constraint (we know the upper bound never binds).

    In fact, the non-negativity constraint is never binding (i.e.\ $\lambda(v) = 0$ a.e.). Suppose there is an interval $(v_1,v_2)$ such that the constraint binds in this interval (that is,  $\lambda(v)>0$ and, consequently, also  $h(v) = 0$ on this interval). We must clearly have that the constraint slacks in this interval, so $\mu[\underline v,v]$ is constant for all  $v\in[v_1,v_2]$. Furthermore, if $v_1>\underline v$, then the constraint must slack in $v\in[v_1-\epsilon,v_2]$ for an $\epsilon$ small enough (which follows from the fact that $E_h(v)$ is strictly decreasing in $[v_1,v_2]$). If $v_1 = \underline{v}$, then $\mu[\underline v,v]$ must be equal to zero in $v\in[v_1,v_2]$, but this contradicts the first-order condition. If $v_1>\underline v$, then we must have that:
    \[
    \frac{\partial}{\partial h}\overline{u}(v_1-\epsilon,h(v_1-\epsilon)) = \mu[\underline v,v_1]>  \mu[\underline v,v_1] - \lambda(v_1+\epsilon),
    \]
    where we simply used that the constraint slacks so $\mu[\underline v,v]$ stays constant in $[v_1-\epsilon,v_1+\epsilon]$ (for a small enough $\epsilon$). However, we also have that:
    \[\frac{\partial}{\partial h}\overline{u}(v_1-\epsilon,h(v_1-\epsilon)) \leq \frac{\partial}{\partial h}\overline{u}(v_1-\epsilon,0)\leq   \frac{\partial}{\partial h}\overline{u}(v_1+\epsilon,0),\]
    where in the first inequality we used that $\overline u$  is concave and in the second one we used that $\partial \overline u(v,0)/\partial h$ is increasing in $v$. We thus reach a contradiction with the fact that the first-order condition is satisfied at $v=v_1+\epsilon$. Thus, the KKT conditions simplify to
    \begin{equation}\label{eq:kkt}
    \frac{\partial}{\partial h}\overline{u}(v,h(v)) = \mu[\underline v,v] \text{ and } E_h(v) d\mu(v) = 0 \text{ a.e.}
    \end{equation}

    Our final step is to show that \eqref{eq:kkt} implies that there is no relaxation gap by considering the concavification in \eqref{eq:relaxed}. We prove that  for any $h$ satisfying \eqref{eq:kkt}, $u(v,h(v)) = \overline{u}(v,h(v))$ a.e. An implication of \eqref{eq:kkt} is that
    \[
        \frac{d}{dv} \frac{\partial}{\partial h} \overline{u}(v,h(v)) = \frac{\partial^2}{\partial v \partial h} \overline{u}(v,h(v)) + \frac{dh}{dv} \cdot \frac{\partial^2}{\partial h^2} \overline{u}(v,h(v)) \geq 0 \text{ a.e.}
    \]
    since $\mu$ is a non-negative measure. Take any $v$ where $u(v,h(v)) < \overline{u}(v,h(v))$. The second term of the above expression is zero ($\overline u$ is affine). Let $(h_B,h_T)$ be the bottom and top of the support point of $\overline{u}(v,h(v))$, respectively, so that
    \begin{equation}\label{eq:h-slope}
        \frac{\partial \overline u(v,h(v))}{\partial h}=\frac{\partial u(v,h_B(v))}{\partial h}=\frac{\partial u(v,h_T(v))}{\partial h}=\frac{u(v,h_T(v))-u(v,h_B(v))}{h_T(v)-h_B(v)}.
    \end{equation}
    Here we are using that $\overline u(v,h)=u(v,h)$ for all $h < \epsilon$ (for a small enough $\epsilon$), so the concavification differs from $u$ only at interior hazard rates. We thus get that
    \[
    \frac{d}{dv}\frac{\partial \overline u(v,h(v))}{\partial h}=\frac{1}{h_T(v)-h_B(v)}\left(\frac{\partial u(v,h_T(v))}{\partial v}-\frac{\partial u(v,h_B(v))}{\partial v}\right).
    \]
    Here we took the derivative of the rhs of \eqref{eq:h-slope}, and then used \eqref{eq:h-slope} to cancel out all the terms with derivatives of $h_T(v)$ and $h_B(v)$, so only the partial derivative with respect to $v$ is left. Finally, using the middle equality of \eqref{eq:h-slope}, we get that:
    \begin{align*}
        \frac{d}{dv}\frac{\partial \overline u(v,h(v))}{\partial h}
        =&\frac{1}{h_T(v)-h_B(v)}\int_{h_B(v)}^{h_T(v)}\left(\frac{\partial^2 u(v,h)}{\partial h^2}+\frac{\partial^2 u(v,h)}{\partial h\partial v}\right)dh < 0
    \end{align*}
    where the inequality follows from
    \begin{equation}\label{eq:supermod-no-h}
    \frac{\partial^2 u(v,h)}{\partial h^2}+\frac{\partial^2 u(v,h)}{\partial h\partial v} = -Q^\prime(v - h) < 0.
    \end{equation}
    Thus, there can be no positive measure of points where $u(v,h) < \overline u(v,h)$ for any $h$ that satisfies \eqref{eq:kkt}. Since $m^\ast$ is absolutely continuous, this means any solution of the relaxed problem \eqref{eq:relaxed} achieves the same value in the original problem:
    \[
        \max_{\mathcal{H}} \int_{\underline{v}}^{\overline{v}} \overline{u}(v,h(v))\ dF^\ast(v)=   \max_{\mathcal{H}} \int_{\underline{v}}^{\overline{v}} {u}(v,h(v))\ dF^\ast(v),
    \]
    and the solution to \eqref{eq:relaxed} is also a maximizer of \eqref{eq:ub}.
\end{proof}

Before moving on to the construction of the uniform segmentation implementing the solution, we prove that the solution satisfies regularity and the constraint binds near $\overline{v}$.

\newpage
\begin{lemma}[Regularity of Optimal $h$]\label{lm:h-opt-reg}
    For any $h \prec h^\ast$ which satisfies \eqref{eq:kkt}, \textup{(i)} $h^\prime(v) < 1$ a.e., and \textup{(ii)} there exists a $\delta > 0$ such that $h(v) = h^\ast(v)$ for almost all $v \in [\overline{v} - \delta, \overline{v}]$.
\end{lemma}

It is then without loss to assume that (i) and (ii) hold exactly everywhere, since measure 0 sets do not matter when $m^\ast$ is absolutely continuous. Then, (i) both implies that $v - h(v)$ is monotone, i.e.\ the canonical segmentation implementing $h$ is supported on (strictly) regular markets only, and that $h$ has bounded total variation.

\begin{proof}
    Take the total derivative of the marginal local information rent with respect to $v$:
    \[
    \frac{d}{dv} \frac{\partial u(v,h(v))}{\partial h} = h^\prime(v) \frac{\partial^2 u(v,h(v))}{\partial h^2} + \frac{\partial^2 u(v,h(v))}{\partial h \partial v} \geq 0.
    \]
    Observe that $u(v,h(v)) = \overline{u}(v,h(v))$ implies $u$ is concave in $h$ at $h(v)$ (Theorem 7.23, \cite{albo06}); (i) then follows from \eqref{eq:supermod-no-h}. For (ii), suppose by way of contradiction that there is an interval $(\overline{v} - \delta,\overline{v})$ such that $E_h(v) > 0$ on this interval ($E_h$ is continuous even if $h$ is not continuous). By \eqref{eq:kkt}, $\frac{\partial}{\partial h} \overline{u}(v,h(v))$ must be constant on this interval, and at most $Q(\overline{v} - \delta)$. But, $h^\ast(v) \rightarrow 0$, meaning $\frac{\partial}{\partial h} \overline{u}(v,h^\ast(v)) \rightarrow Q(\overline{v})$. This means for all $v$ sufficiently close to $\overline{v}$, $h(v) > h^\ast(v)$, which means the majorization constraint is not satisfied. We thus reach a contradiction.
\end{proof}

\begin{proposition}[Implementability of Optimal $h$]\label{prop:canon-seg}
    There exists a uniform segmentation $\sigma$ implementing the solution to \eqref{eq:ub}.
\end{proposition}

\begin{proof}
    Consider a continuum of markets indexed by $m \in [0,1]$. The support is determined by a function $s: V \rightarrow [0,1]$ such that $v\in\mathrm{supp}(m)$ if and only if $m\leq s(v)$. The distribution is absolutely continuous in market $m$ at all $v$ such that $m < s(v)$ with density $f_m(v)$:
    \begin{equation}\label{dens}
        f_m(v) = \frac{D_m(v)}{h(v)},
    \end{equation}
    and there is possibly an atom of size $a_m(v)$ at $v$ at market $m = s(v)$:
    \begin{equation}\label{atom}
        a_m(v) = \frac{\Delta_{m}(v) D_{m}(v)}{h(v)}.
    \end{equation}
    Finally, $s(v)$ satisfies \eqref{eq:h-sigma}, which in this case can be written more simply as follows:
    \begin{equation}\label{eq:s}
        \int^{1}_{s(v)}\Delta_{m}(v)D_m(v)\ dm = E_h(v).
    \end{equation}

    The proof proceeds in two steps. First, we show that given a solution to \eqref{eq:s}, the associated segmentation given by \eqref{dens}-\eqref{atom} aggregates to $m^\ast$. That is, for all $v$,
    \begin{equation}\label{eq:feas-s}
    D^\ast(v) = \int_0^1 D_m(v)\ dm \iff \int_v^{\overline{v}} D^\ast(t)\ dt = \int_v^{\overline{v}} \int_0^1 D_m(t)\ dm\ dt.
    \end{equation}
   Here, the equivalence simply states that two functions are the same if and only if their
integrals are the same.  As before, we have that for any $m$,
    \[
    \int_v^{\overline{v}} D_m(t)\ dt = \int_v^{\overline{v}} h(t)\ dF_m(t) + \Delta_m(v) D_m(v).
    \]
    The second term is only nonzero when $m > s(v)$ (or, equivalently,  $v \notin \supp(m)$), so:
    \[
    \int_v^{\overline{v}} \int_0^1 D_m(t)\ dm\ dt = \int_0^1 \int_v^{\overline{v}} h(t)\ dF_m(t)\ dm + \int_{s(v)}^1 \Delta_m(t) D_m(t)\ dm,
    \]
    where we apply Fubini's theorem to exchange the integration order.
    Now, using \eqref{eq:s} to replace the second term on the rhs, we get
    \[
    \int_v^{\overline{v}} \int_0^1 D_m(t)\ dm\ dt = \int_0^1 \int_v^{\overline{v}} h(t)\ dF_m(t)\ dm + \int_v^{\overline{v}} D^\ast(t)\ dt - \int_v^{\overline{v}} h(t)\ dF^\ast(t).
    \]
    Let $\widehat{v}$ be the largest value such that the lhs of \eqref{eq:feas-s} fails; note that the rhs of \eqref{eq:feas-s} still holds for $v = \widehat{v}$. Plugging this into the above equation, we get
    \[
    \int_{\widehat{v}}^{\overline{v}} D^\ast(t)\ dt = \int_0^1 \int_{\widehat v}^{\overline{v}} h(t)\ dF_m(t)\ dm + \int_{\widehat v}^{\overline{v}} D^\ast(t)\ dt - \int_{\widehat v}^{\overline{v}} h(t)\ dF^\ast(t)
    \]
    which contradicts that the lhs of \eqref{eq:feas-s} fails at $\widehat{v}$.

    The second part of the proof is to show that \eqref{eq:s} has a solution. We do this by taking a sequence $\{h_n\}$ approximating $h$, finding a solution $s_n$ for each $h_n$, and argue that the sequence $\{s_n\}$ converges. We take $h_n$ to be a function with the following properties:
    \[
    h_n \text{ is real-analytic, }\epsilon_n \leq h_n(v) \leq \overline{h}(v)\text{, and \eqref{eq:maj-0} holds for all } v \in [\underline{v},\overline{v} - \delta],
    \]
    where $\epsilon_n \rightarrow 0$ and $h_n(v) = h^\ast(v)$ for $v \geq \overline{v} - \delta$. Since the space of real-analytic functions is dense in the space of functions with bounded total variation, we can construct a sequence $\{h_n\}$ converging pointwise a.e.\ to $h$.\footnote{Technically, the sequence converges to $h$ in the $L^p$-norm, and by the Riesz-Fisher theorem we can then extract a pointwise convergent subsequence.}

We construct $s(v)$ inductively. Consider $v_2$ such that $s(v)$ satisfies \eqref{eq:s} for all  $v\in[v_2,\overline v]$ and $s^\prime(v_2)=0$. We construct $s(v)$ in a new interval $[v_1,v_2]$ satisfying feasibility and such that $s^\prime(v)=0.$ We consider two cases, and show that we move between cases only when $E$ has a strict inflection point. Since    $h_n$ and $h^\ast$ are real-analytic,  $E_{h_n}$ has a finite number of strict inflection points.  Hence, the process finishes after a finite number of steps. Finally, note that $s(v) = 1$ for $v \geq \overline{v} - \delta$, so we start the induction with Case 1.

\textbf{(Case 1: $E_h^{\prime\prime}(v) \geq 0$  for all $v$ in some neighborhood $[v_2-\epsilon,v_2]$)} Define, for each $m$,
\[
\widehat{D}_m(v) = D_m(v_2) \exp\left(\int_v^{v_2} \frac{1}{h_n(t)}\ dt\right).
\]
Since $h_n$ is strictly bounded away from 0, $\frac{1}{h_n}$ is integrable and this is well-defined. These survival functions satisfy the constant hazard rate requirement. Now, define $s_n(v)$ by the following differential equation:
\begin{equation}\label{eq:ode-s}
    \frac{ds_n}{dv} = \frac{1}{\widehat{D}_{s_n(v)}(v)} \cdot \frac{d^2 E_{h_n}}{dv^2}.
\end{equation}
This is an ordinary differential equation where the rhs is continuous in $s$ (recall that $h_n$ is real-analytic), so by Carath\'{e}odory's existence theorem, a solution exists. It is clear that in this interval, $s$ is increasing. Furthermore, \eqref{eq:ode-s} is derived from \eqref{eq:s} by taking the second derivative and setting $\Delta_m(v) = 0$ (since $s_n$ is increasing), and hence \eqref{eq:s} is satisfied.

\textbf{(Case 2: $E_h^{\prime\prime}(v) < 0$   for all $v$ in some neighborhood $[v_2-\epsilon,v_2]$)}
For each $m$, let $\widehat{v}(m)$ be the lowest inverse of $s_n(v)$ in $[v_2,\overline{v}]$:
\[
\widehat{v}(m) \equiv \inf\big\{v \in [v_2,\overline{v}] \mid s_n(v) \geq m\big\}.
\]
For all $m \geq s_n(v_2)$, this value exists, since $s_n(\overline{v}) = 1$. We define $s_n(v)$ so that
\begin{equation}\label{eq:int-s}
    \int_{s_n(v)}^1 \big(\widehat{v}(m) - v\big) D_m(\widehat{v}(m))\ dm = E_{h_n}(v).
\end{equation}
Note that this is just \eqref{eq:s} under the assumption that $s_n(v)$ is decreasing in $[v,v_2]$. Let $v_1$ be the largest value less than $v_2$ such that the solution to \eqref{eq:int-s} is decreasing in $[v_1,v_2]$.  Taking the second derivative of \eqref{eq:s}, and evaluating at $v_1$ (we have $s_n'(v_1)=0$) we get that
\[
E_{h_n}^{\prime\prime}(v_1) = -s_n^{\prime\prime}(v_1) \big(\widehat{v}(s_n(v_1)) - v_1\big) D_m(\widehat{v}(s_n(v_1))) \geq 0
\]
meaning that $[v_1,v_2]$ is at least as large as the gap between the strict inflection points of $E_{h_n}$. Thus, we move back to Case 1, and after a finite number of steps, the process ends.

Finally, we note that since $h^\prime(v) < 1$ everywhere, when $s_n$ is increasing,
\[
\frac{ds_n}{dv} = \frac{1}{\widehat{D}_{s(v)}(v)} E^{\prime\prime}_{h_n}(v) < \frac{1}{D^\ast(\overline{v} - \delta)}\left(\max_v \left[2 f^\ast(v) + \frac{df^\ast(v)}{dv} \overline{h}(v)\right]\right).
\]
Hence, all $s_n$ are bounded in $[0,1]$ with uniform upper bound on $s_n^\prime(v)$, which in turn means they have uniformly upper bounded total variation. Thus, by Helly's Selection Theorem, there is a subsequence of $s_n$ converging pointwise.
\end{proof}

Together, these results establish Proposition \ref{prop:unif-seg-ub}.
\end{proof}

As a corollary of Proposition \ref{prop:sup-max}, we get necessary and
sufficient conditions for $h^\ast$ to solve \eqref{eq:ub}, extending
Corollary \ref{cor:no-seg2} to the general setting.

\begin{corollary}[Zero Segmentation]
\label{cor:no-seg} Zero segmentation is optimal if and only if both \textup{%
(i)} $\frac{\partial}{\partial h} u(v,h^\ast(v))$ is non-decreasing and 
\textup{(ii)} $u(v,h^\ast(v)) = \overline u(v,h^\ast(v))$ hold a.e.
\end{corollary}

\begin{proof}[Proof of Theorem \ref{thm:main}]
    Theorem \ref{thm:main} consists of two parts. The first, that the value of segmentation is characterized by \eqref{eq:main-thm}, follows from Proposition \ref{prop:unif-seg-ub} after we rule out non-uniform segmentations. Any segmentation that is supported only on regular markets generates consumer surplus
    \[
    \int_{\Delta V} \int_{\underline{v}}^{\overline{v}} u(v,h_m(v))\ d\Sigma_v(m)\ dF^\ast(v) \leq \int_{\underline{v}}^{\overline{v}} \int_{\Delta V} \overline{u}(v,h_m(v))\ d\Sigma_v(m)\ dF^\ast(v).
    \]
    But, by Jensen's inequality,
    \[
    \int_{\underline{v}}^{\overline{v}} \int_{\Delta V} \overline{u}(v,h_m(v))\ d\Sigma_v(m)\ dF^\ast(v) \leq \int_{\underline{v}}^{\overline{v}} \overline{u}(v,h_\sigma(v))\ dF^\ast(v),
    \]
    with strict inequality whenever $h_m \neq h_\sigma$ a.e. (indeed, equality in Jensen would require \(h_\sigma(v)\) to lie on an affine face of \(\overline u(v,\cdot)\), which is ruled out at an optimum by the no-relaxation-gap argument). But, by Propositions \ref{prop:avg-h}-\ref{prop:unif-seg-ub}, the maximum possible value of the rhs is achievable with a uniform segmentation, so any non-uniform segmentation achieves lower consumer surplus than the optimal uniform segmentation. This proves that the optimal segmentation attains value \eqref{eq:main-thm}. 
    
    The second part of the theorem requires us to rule out segmentations supported on irregular markets that nonetheless achieve the same value. To do this, take any irregular market, and decompose it via Lemma \ref{lm:csir} to weakly regular markets (of which there must be a positive measure). Our first step is to show that
    \begin{equation}\label{eq:sup-cont}
        \max_{\sigma \in \operatorname{MPS}(m)} \int_{x \in \Delta V} U(x)\ d\sigma(x) \geq \sup_{h \prec h_m} \int_{\underline{v}}^{\overline{v}} \overline{u}(v,h(v))\ dF_m(v).
    \end{equation}
    If $m$ is absolutely continuous, this is immediate because it follows from \eqref{eq:main-thm}. To accommodate markets with atoms, we first prove that the rhs of \eqref{eq:sup-cont} is lsc while the lhs is usc in $m$.

For the rhs, we assume all $h\in\mathcal{H}$ are lsc (this convention obviously does not matter when $m$ is continuous). Since $u$ is increasing in $h$, the objective function is lsc in $F_m$ (Theorem 15.3, \cite{albo06}). On the other hand, the set 
\[ \{h\in\mathcal{H}:h\prec h_m\}\]
is lower hemicontinuous in $m$. Lower hemicontinuity is straightforward because any $\widehat h(v)$ that is pointwise smaller than $h(v)$, with $h\prec h_m$, satisfies $\widehat h\prec h_{m^\prime}$ for all $m^\prime$ close enough to $m$. Hence, following the maximum theorem, the value is lsc in $m$.
 
For the lhs of \eqref{eq:sup-cont}, note that, under our tie-breaking rules, $U(m)$ is bounded and usc (see \cite{kahy23}), meaning that (Theorem 15.3, \cite{albo06})
    \[
    \int_{x \in \Delta V} U(x)\ d\sigma(x)
    \]
    is itself usc in $\sigma$. Furthermore, by standard results, the mapping $m \mapsto \operatorname{MPS}(m)$ is nonempty, compact-valued, and continuous. Hence, by Berge's theorem, the lhs is usc.

    Now, if $m$ is only weakly regular, there exists a positive measure subinterval $[v_1,v_2]$ (possibly with atoms) such that $\frac{\partial}{\partial h} \overline{u}(v,h_m(v))$
    is strictly decreasing (if $u(v,h_m(v)) = \overline{u}(v,h_m(v))$, this follows from \eqref{eq:supermod-no-h}; else, see \eqref{eq:h-slope}). Hence, $h_m$ does not satisfy the first-order condition, meaning it does not solve \eqref{eq:relaxed}.\footnote{When $m$ has atoms, the KKT system may not characterize the solution, but the FOC is still necessary; see \eqref{eq:d-gaph} and the discussion that follows.} We thus have that:
    \[
    \int_{\underline{v}}^{\overline{v}} \overline{u}(v,h_m(v))\ dF_m(v)<  \sup_{h \prec h_m} \int_{\underline{v}}^{\overline{v}} \overline{u}(v,h(v))\ dF_m(v)\leq  \max_{\sigma \in \operatorname{MPS}(m)} \int_{x \in \Delta V} U(x)\ d\sigma(x).
    \]
    But, the lhs is weakly larger than the consumer surplus in $m$. So, there is a further segmentation of $m$ which strictly improves consumer surplus.
\end{proof}

\begin{proof}[Proof of Proposition \ref{prop:cmv}]
We first note that for any $h$,
\begin{equation}\label{refzxc}
\frac{\partial u(v,h)}{\partial h}\geq0\implies \frac{\partial^2 u(v,h)}{\partial h^2}<0.
\end{equation}
This follows from the fact that
\begin{equation}\label{cc8dc}
\frac{\partial u(v,h)}{\partial h}\geq0 \iff \frac{1}{h}\geq \frac{Q^{\prime}(v-h)}{Q(v-h)}
\end{equation}
and
\[
h\left(Q^{\prime\prime}(v-h)-\frac{2Q^{\prime}(v-h)}{h}\right)\leq h\left(Q^{\prime\prime}(v-h)-2\frac{(Q^{\prime}(v-h))^2}{Q(v-h)}\right)\leq-h\frac{(Q^{\prime}(v-h))^2}{Q(v-h)},
\]
where the first inequality uses \eqref{cc8dc} while the second one follows from \eqref{eq:conv-mc}. We thus obtain \eqref{refzxc}. This implies that $u(v,h)$ is strictly quasi-concave in $h$. Furthermore, we have that:
\[\frac{\partial^2 u(v,h)}{\partial h\partial v}=h\left(\frac{Q^\prime(v-h)}{h}-Q^{\prime\prime}(v-h)\right)\geq 0,\]
where the inequality follows as before from \eqref{cc8dc} and \eqref{eq:conv-mc}. This implies that $\overline h(v)$ is increasing, which, in turn, implies that $\overline h(v)$ and $h^\ast(v)$ cross at most once (where we use that $h^\ast$ is decreasing for this last implication). Hence, we have that $u(v,h(v)) = \overline u(v,h(v))$ for $h(v)$ defined as in \eqref{fcx}. Finally, we show that $\partial u(v,h(v))/\partial h$ is increasing in $v$. For all $v\leq \widehat v$ we have that the derivative is zero; for all $v\geq \widehat v$ we have that the derivative is positive. Finally, we have that the derivative is also increasing  for all $v\geq \widehat v$:
\begin{align*}
    \frac{d}{dv}\left(\frac{\partial u(v,h^\ast(v))}{\partial h}\right)=&\left(1-\frac{dh^\ast(v)}{dv}\right) \left(\frac{\left(1-2 \frac{dh^\ast(v)}{dv}\right)}{\left(1-\frac{dh^\ast(v)}{dv}\right) } Q'(v-h^\ast(v))-h^\ast(v) Q''(v-h^\ast(v))\right) \\
    &>\left(1-\frac{dh^\ast(v)}{dv}\right) \left(  Q'(v-h^\ast(v))-h^\ast(v) Q''(v-h^\ast(v))\right)>0,
\end{align*}
where the first inequality follows from the fact that $h^\ast(v)$ is decreasing, so the fraction is larger than 1, while the second inequality follows from  \eqref{cc8dc} and \eqref{eq:conv-mc}. Thus, $\partial u(v,h(v))/\partial h$ is increasing in $v$. Following Proposition \ref{prop:sup-max}, this is the solution to \eqref{eq:main-thm}.
\end{proof}
\begin{proof}[Proof of Theorem \ref{thm:conv-seg}]
The only thing left is to show that the canonical segmentation is gapless and atomless. By the proof of Proposition \ref{prop:canon-seg}, this is equivalent to $s$ being increasing, or equivalently, $E_h$ being convex. By \eqref{eq:s},
\[
\frac{d}{dv} E_h(v) = \big[h(v) - h^\ast(v)\big] f^\ast(v).
\]
For $v \geq \widehat{v}$, $h(v) = h^\ast(v)$; otherwise, $h(v) < h^\ast(v)$ with $h$ increasing and $h^\ast$ decreasing. So, it is immediate that $E_h$ is convex.
\end{proof}

\begin{proof}[Proof of Proposition \ref{prop:com}] We use the conditions in Corollary \ref{cor:no-seg} to know when zero segmentation is optimal. First, we note that  $u(v,h^\ast(v))=\overline u(v,h^\ast(v))$ if and only if 
\[h^\ast(v)\leq \overline h(v)=\frac{\gamma-1}{\gamma}v.\]
Since $h(v)$ is increasing in $\gamma$, the condition becomes less restrictive as $\gamma$ increases. Second, we note that $\partial u(v,h^\ast(v))/\partial h$ is increasing in $v$ if and only if:
\begin{equation}\label{dcxz}
(\gamma-1)\left(v-\frac{\partial h^\ast(v)}{\partial v}(2v- h^\ast(v) )\right)+ \left(\frac{\partial h^\ast(v)}{\partial v}-1\right)h^\ast(v)\geq0.
\end{equation}
Finally, we have that $\partial h^\ast(v)/\partial v\leq1$ so the second term is negative (otherwise, this is an irregular market, so zero segmentation is not optimal due to Theorem \ref{thm:main}). Thus, if at some $v$, \eqref{dcxz} is satisfied with equality, then \eqref{dcxz} will not be satisfied also for all lower values of $\gamma$. Hence, among regular markets, \eqref{dcxz} also becomes less restrictive as $\gamma$ increases.
\end{proof}

We now provide the proofs of the results in Section \ref{sec:gen-exts}. We
consider the following problem: 
\begin{equation}  \label{dcxz4}
\begin{split}
\overline W\equiv &\max_{\sigma\in \operatorname{MPS}(m^\ast)}\int_{\Delta
V}\int_{\underline v}^{\overline v} w(v,h_m(v))\ dF_m(v)\ d\sigma(m) \\
&\text{subject to: every }m\in\mathrm{supp}(\sigma)\text{ is regular}.
\end{split}%
\end{equation}
This is the same as \eqref{reg:maxim}, but replaced $u(v,h)\to w(v,h)$.
Lemma \ref{lm:csir} will be satisfied in all the extensions we consider, so
considering this formulation will indeed be justified. We assume that $w$ is
not too supermodular relative to the concavity: 
\begin{equation}
\frac{\partial ^{2}w(v,h)}{\partial h^{2}}+\frac{\partial ^{2}w(v,h)}{%
\partial h\partial v}\leq 0.  \label{eq:smc}
\end{equation}%
This is an economically substantive assumption---it implies that higher
values are associated with only moderately higher inverse hazard rates (or
lower inverse hazard rates)---and is used to guarantee the solution to %
\eqref{eq:ub} is regular. Second, we assume that there exists $\epsilon$
such that for all $h\in(0,\epsilon)$: 
\begin{equation}  \label{condw2}
\frac{\partial w(v,h)}{\partial h}\text{ is non-decreasing in }v\text{ and
is greater than }\frac{w(v,h^{\prime })}{h^{\prime }}\text{ for all }h^{\prime
}\in[\epsilon,\infty).
\end{equation}
This assumption rules out corner cases at \(h=0\) in the relaxed problem. In particular, it ensures that whenever the concavification differs from the original objective, this happens away from the boundary. This is used immediately after \eqref{eq:h-slope} in our main argument.

\begin{theorem}[Generalized Upper Bound]
\label{thm:gen} Under \eqref{eq:smc}-\eqref{condw2}, the maximum of %
\eqref{dcxz4} is given by: 
\begin{equation}
\overline W=\max_{h\prec h^{\ast }}\ \int_{\underline{v}}^{\overline{v}%
}w(v,h(v))\ dF^{\ast }(v).  \label{eq:geng1}
\end{equation}
Furthermore, every segmentation achieving this value is a uniform
segmentation implementing some $h$ which solves the rhs of \eqref{eq:geng1}.
\end{theorem}

The theorem generalizes our main result to objective functions that do not
consist of maximizing the local informational rents. The theorem simply
identifies the conditions of $u$ we used to prove Theorem \ref{thm:main}, so
it does not involve any new analysis. We now provide the proofs of the
results in the main text.

\begin{proof}[Proof of Proposition \ref{ext:adv}]
Following standard techniques, the seller's supply is:
\[Q(\phi_m(v))=v-\frac{d\tau(v)}{dv}\frac{1-F_m(v)}{f_m(v)}.\]

Hence, the expected consumer surplus,
\[\int_{\underline v}^{\overline v}\frac{d\tau(v)}{dv}\frac{1-F_m(v)}{f_m(v)}\left(v-\frac{d\tau(v)}{dv}\frac{1-F_m(v)}{f_m(v)}\right)dF_m(v).\]
We thus obtain that the consumer surplus is of the form \eqref{dcxz4} with $w(v,h)$ given by \eqref{adv:5}. Lemma \ref{lm:csir} is clearly satisfied as we can apply the same arguments. Furthermore,
\[
\frac{\partial ^{2}w(v,h)}{\partial h^{2}}+\frac{\partial
^{2}w(v,h)}{\partial h\partial v} = \frac{d\tau(v)}{dv}\left( 1+\frac{\frac{vd^2\tau(v)}{dv^2}}{\frac{d\tau(v)}{dv}}-2\frac{d\tau(v)}{dv}\right)
   -4 h(v)\frac{d\tau(v)}{dv}  \frac{d^2\tau(v)}{dv^2}<0
\]
and
\[
\frac{\partial w(v,h)}{\partial h}\mid_{h=0} = v\frac{d\tau(v)}{dv}.
\]
Thus, \eqref{eq:smc}-\eqref{condw2} are satisfied.
\end{proof}

\begin{proof}[Proof of Proposition \ref{ext:sw}]
We can verify that for all $\lambda\geq -1$, Lemma \ref{lm:csir} holds and $w(v,h)$ satisfies:
\[
\frac{\partial ^{2}w(v,h)}{\partial h^{2}}+\frac{\partial
^{2}w(v,h)}{\partial h\partial v} =-(1+\lambda)Q^{\prime}(v-h) \text{ and } \frac{\partial w(v,h)}{\partial h}\mid_{h=0}=Q(v).
\]
Thus, \eqref{eq:smc}-\eqref{condw2} are satisfied.
\end{proof}

\newpage 
\bibliographystyle{econometrica}
\bibliography{references.bib}

\end{document}